\def \VersionLong {}
\ifdefined\VersionLong
	\newcommand{\LongVersion}[1]{\ifdefined\VersionWithComments{\color{red!40!black}#1}\else#1\fi}
	\newcommand{\ShortVersion}[1]{\ifdefined\VersionWithComments{\color{black!40}#1}\fi}
\else
	\newcommand{\LongVersion}[1]{\ifdefined\VersionWithComments{\color{black!40}#1}\fi}
	\newcommand{\ShortVersion}[1]{\ifdefined\VersionWithComments{\color{red!40!black}#1}\else#1\fi}
\fi

\ifdefined\LaTeXdiff
	\let\VersionWithComments\undefined
	\let\WithReply\undefined
\fi
\ShortVersion{
	\documentclass[preprint,12pt,authoryear]{elsarticle}
}
\LongVersion{
	\documentclass[a4paper,10pt]{article}
	\sloppy
}

\usepackage[utf8]{inputenc}
\usepackage[english]{babel}

\usepackage[ruled,vlined,linesnumbered]{algorithm2e}
\SetKwInOut{Input}{input}
\SetKwInOut{Output}{output}
\SetKwInOut{Variables}{variables}

\def\HiLi{\leavevmode\rlap{\hbox to \hsize{\color{yellow!50}\leaders\hrule height .8\baselineskip depth .5ex\hfill}}}

\usepackage{subcaption}

\usepackage{paralist} %

\usepackage{amsmath} %
\usepackage{amssymb} %

\usepackage{amsthm}
\theoremstyle{definition}
\newtheorem{definition}{Definition}
\newtheorem{example}{Example}

\theoremstyle{plain}
\newtheorem{lemma}{Lemma}
\newtheorem{proposition}{Proposition}
\newtheorem{theorem}{Theorem}
\theoremstyle{remark}
\newtheorem{remark}{Remark}

\ShortVersion{
\theoremstyle{definition}

}
\LongVersion{
	\usepackage[backend=biber,backref=true,style=alphabetic,url=false,doi=true,defernumbers=true,sorting=anyt,maxnames=99]{biblatex} %
	\addbibresource{PIPTA.bib}

	\renewbibmacro*{doi+eprint+url}{%
		\iftoggle{bbx:doi}
			{\color{black!40}\footnotesize\printfield{doi}}
			{}%
		\newunit\newblock
		\iftoggle{bbx:eprint}
			{\usebibmacro{eprint}}
			{}%
		\newunit\newblock
		\iftoggle{bbx:url}
			{\usebibmacro{url+urldate}}
			{}%
	}

}
\ifdefined\VersionWithComments
	\usepackage{draftwatermark}
	\SetWatermarkText{draft}
	\SetWatermarkScale{15}
	\SetWatermarkColor[gray]{0.9}
\fi
\usepackage[svgnames,table]{xcolor}
\definecolor{darkblue}{rgb}{0.0,0.0,0.6}
\definecolor{darkgreen}{rgb}{0, 0.5, 0}
\definecolor{darkpurple}{rgb}{0.7, 0, 0.7}
\definecolor{violetcurie}{RGB}{115,26,67}
\definecolor{forestgreen}{rgb}{0.13,0.54,0.13}
\definecolor{darkblue}{rgb}{0, 0, 0.7}

\usepackage[
		colorlinks=true,
	\ifdefined \VersionWithComments
		pagebackref=true,
	\fi
		citecolor=darkgreen,
		linkcolor=darkblue,
		urlcolor=darkpurple,
	]{hyperref}

\usepackage[capitalise,english,nameinlink]{cleveref} %
\crefname{line}{\text{line}}{\text{lines}} %
\crefname{figure}{\text{Figure}}{\text{Figures}} %

\usepackage{tikz}
\usetikzlibrary{arrows,automata,positioning}
\tikzstyle{every node}=[initial text=]
\tikzstyle{location}=[rectangle, rounded corners, minimum size=12pt, draw=black, fill=blue!10, inner sep=2pt]
\tikzstyle{probchoice}=[circle, minimum size=5pt, draw=black, fill=black, inner sep=0pt]
\tikzstyle{probedge}=[densely dotted]
\tikzstyle{final}=[double]
\tikzstyle{point}=[probchoice]

\newcommand{\AngleA}[3]
	{\begin{scope}
		\path [clip] (#1.center) -- (#2.center) -- (#3.center);
		\fill [blue!20,draw=black,thin] (#1) circle (3mm);

		\node[probchoice] at (#1.center) {};
	\end{scope}}

\definecolor{cv1}{rgb}{1, 0, 0}
\definecolor{cv2}{rgb}{0, 1, 0}
\definecolor{cv3}{rgb}{0, 0, 1}
\definecolor{cv4}{rgb}{1, 1, 0}
\definecolor{cv5}{rgb}{1, 0, 1}
\definecolor{cv6}{rgb}{0, 1, 1}
\definecolor{cv7}{rgb}{0.8, 0.6, 0.4}
\definecolor{cv8}{rgb}{0.5, 0.5, 1}
\definecolor{cv9}{rgb}{0.55, 0.75, 0.35}
\definecolor{cv10}{rgb}{1, 0.6, 0.1}
\definecolor{cv11}{rgb}{0.6, 0.7, 0.8}
\definecolor{cv12}{rgb}{0.2, 0.5, 0.9}
\definecolor{cv13}{rgb}{0.5, 0.9, 0.2}
\definecolor{cv14}{rgb}{1, 0.3, 0.5}
\definecolor{cv15}{rgb}{0.7, 0.7, 0.7}

\tikzstyle{location0}=[location, fill=cv1!50]
\tikzstyle{location1}=[location, fill=cv2!50]
\tikzstyle{location2}=[location, fill=cv3!50]
\tikzstyle{location3}=[location, fill=cv4!50]
\tikzstyle{location4}=[location, fill=cv5!50]
\tikzstyle{location5}=[location, fill=cv6!50]
\tikzstyle{location6}=[location, fill=cv7!50]

\newcommand{\cellHeader}[1]{\cellcolor{blue!20}\textbf{#1}}

\ifdefined \VersionWithComments
	\usepackage{marginnote}
	\newcommand{\marginX}{\marginnote{\huge{\quad\quad\textbf{!}\quad\quad}}}
	\newcommand{\avirer}[1]{{\color{black!40}#1]}}
	\newcommand{\ea}[1]{\mbox{}{\color{green!50!black}\marginX{}\textbf{[\'Etienne}: #1]}}
	\newcommand{\bd}[1]{\mbox{}{\color{orange}\marginX{}\textbf{[Benoît}: #1]}}
	\newcommand{\pf}[1]{\mbox{}{\color{yellow!50!black}\marginX{}\textbf{[Paulin}: #1]}}
	\newcommand{\instructions}[1]{{\color{red}\marginX{}\textbf{[Instructions: ``#1'']}}}
	\newcommand{\todo}[1]{\mbox{}{\color{red}{\marginX{}\textbf{TODO}\ifx#1\\\else:\ \fi #1}}} %
\else
	\newcommand{\instructions}[1]{}
	\newcommand{\avirer}[1]{}
	\newcommand{\ea}[1]{}
	\newcommand{\bd}[1]{}
	\newcommand{\pf}[1]{•}
	\newcommand{\todo}[1]{}
\fi

\newcommand{\init}{_0}

\newcommand{\A}{\ensuremath{\mathcal{A}}}
\newcommand{\Actions}{\Sigma}
\newcommand{\action}{a}

\newcommand{\BFalse}{\text{false}}
\newcommand{\C}{C}
\newcommand{\Cinit}{\C\init}
\newcommand{\Clock}{X} %
\newcommand{\ClockCard}{H} %
\newcommand{\clock}{x} %
\newcommand{\clockval}{w} %
\newcommand{\CombiInt}{\mathit{FS}}
\newcommand{\PIPCombiInt}{\mathit{CombFS}}
\newcommand{\compOp}{\bowtie}
\newcommand{\correspondance}{\delta}

\newcommand{\DistPTA}{\upsilon}
\newcommand{\distributionTPS}{\eta} %
\newcommand{\DistMDP}{\iota}

\newcommand{\ForbidGuard}{\textit{ForbidD}}
\newcommand{\grandn}{{\mathbb N}}
\newcommand{\grandq}{{\mathbb Q}}
\newcommand{\grandqplus}{\grandq_{+}} %
\newcommand{\grandr}{{\mathbb R}}
\newcommand{\grandrplus}{\grandr_{+}} %
\newcommand{\grandz}{{\mathbb Z}}
\newcommand{\guard}{g}

\newcommand{\EF}{\ensuremath{\mathsf{EF}}}
\newcommand{\EFsynth}{\ensuremath{\mathsf{EFsynth}}}
\newcommand{\EFuniv}{\ensuremath{\mathsf{EFuniv}}}

\newcommand{\DistIMDP}{I}
\newcommand{\DistIPTA}{\Upsilon}

\newcommand{\K}{K}
\newcommand{\Kcons}{\ensuremath{K_\textit{cons}}}
\newcommand{\Kreach}{\ensuremath{K_\textit{reach}}}
\newcommand{\KTrue}{\top}
\newcommand{\KFalse}{\bot}
\newcommand{\loc}{l} %
\newcommand{\locinit}{\loc\init}
\newcommand{\Loc}{L} %
\newcommand{\LocGoal}{G} %
\newcommand{\lterm}{\mathit{aft}}
\newcommand{\Param}{\ensuremath{\Gamma}} %
\newcommand{\param}{\ensuremath{\gamma}} %
\newcommand{\ParamCard}{M} %
\newcommand{\pedge}{e}
\newcommand{\plterm}{\mathit{paft}}
\newcommand{\pval}{v} %
\newcommand{\PZones}{\ensuremath{\mathcal{Z}}}
\newcommand{\RelSimMDP}{\mathcal{R}_M}

\newcommand{\RelSimPTA}{\mathcal{R}_P}
\newcommand{\resets}{\rho}

\renewcommand{\S}{S}

\newcommand{\somelocs}{G} %
\newcommand{\state}{s} %
\newcommand{\States}{S}

\newcommand{\PassedStates}{\mathsf{Passed}}
\newcommand{\IncStates}{\mathsf{Inc}}

\newcommand{\stateinit}{s\init}

\newcommand{\symbstate}{\ensuremath{\mathbf{s}}} %
\newcommand{\SymbStates}{\ensuremath{\mathbf{S}}} %
\newcommand{\symbstateinit}{\symbstate\init} %
\newcommand{\timelapse}[1]{#1^\nearrow}
\newcommand{\TOut}{T_\textit{out}}

\newcommand{\TransitionsMDP}{T} %
\newcommand{\wv}[2]{#1|#2} %

\newcommand{\reconstruct}{\ensuremath{\mathsf{Reconstruct}}}

\newcommand{\makeNonDet}{\ensuremath{\mathsf{makeNonDet}}}
\newcommand{\makeAcc}{\ensuremath{\mathsf{makeAcc}}}

\newcommand{\projectP}[1]{\ensuremath{#1{\downarrow_{\Param}}}}
\newcommand{\reset}[2]{\ensuremath{[#1]_{#2}}}
\newcommand{\valuate}[2]{\ensuremath{#2(#1)}}

\newcommand{\rcpFMax}{\mathit{f\_max}} %
\newcommand{\rcpFMin}{\mathit{f\_min}} %
\newcommand{\rcpSMax}{\mathit{s\_max}} %
\newcommand{\rcpSMin}{\mathit{s\_min}} %

\newcommand{\TransitionsPTA}[0]{ {prob} } %
\newcommand{\TransitionsIPTA}[0]{\mathbb{I}}
\newcommand{\StepsIPTA}[0]{ {\Rightarrow} } %

\newcommand{\dist}{\mathsf{Dist}}

\newcommand{\IntDist}{\mathsf{IntDist}}

\newcommand{\PROBTA}{\normalfont\ensuremath{\mathbb{P}}TA}
\newcommand{\IPROBTA}{\normalfont I\PROBTA}
\newcommand{\PIPROBTA}{\normalfont PI\PROBTA}
\newcommand{\PROBTAs}{\normalfont \PROBTA{}s}
\newcommand{\IPROBTAs}{\normalfont \IPROBTA{}s}
\newcommand{\PIPROBTAs}{\normalfont \PIPROBTA{}s}

\newcommand{\IMDP}{IMDP}
\newcommand{\MDP}{MDP}
\newcommand{\IMDPs}{\IMDP{}s}

\newcommand{\PIPTAsynth}{\ensuremath{\mathsf{ConstSynth}}}
\newcommand{\PIPTAreachsynth}{\ensuremath{\mathsf{ConstEFSynth}}}

\newcommand{\probta}{\ensuremath{\mathcal{P}}}
\newcommand{\iprobta}{\ensuremath{\mathcal{IP}}}
\newcommand{\iprobtainfzero}{\ensuremath{\iprobta_{\infty,0}}}

\newcommand{\piprobta}{\ensuremath{\mathcal{PIP}}}

\newcommand{\imdp}{\ensuremath{\mathcal{IM}}}
\newcommand{\mdp}{\ensuremath{\mathcal{M}}}

\newcommand{\TPS}{\mathcal{T}}

\usepackage[framemethod=TikZ]{mdframed}
\newcommand{\defProblem}[3]
{%
\begin{mdframed}[roundcorner=3pt,backgroundcolor=blue!7,linecolor=blue!70,linewidth=2]
	\textbf{#1 problem:}\\
		\textsc{Input}: #2\\
		\textsc{Problem}: #3
\end{mdframed}
}

\ifdefined \VersionWithComments
 	\definecolor{colorok}{RGB}{80,80,150}
\else
	\definecolor{colorok}{RGB}{0,0,0}
\fi

\newcommand{\eg}{\textcolor{colorok}{e.\,g.,}\xspace}
\newcommand{\ie}{\textcolor{colorok}{i.\,e.,}\xspace}
\newcommand{\st}{\textcolor{colorok}{s.\,t.}\xspace}

\newcommand{\wrt}{\textcolor{colorok}{w.r.t.}\xspace}

\newcommand{\ourtitle}{Consistency in Parametric Interval Probabilistic Timed Automata}

\newcommand{\ourabstract}{%
\begin{abstract}
We propose a new abstract formalism for probabilistic timed systems,
Parametric Interval Probabilistic Timed Automata, based on an
extension of Parametric Timed Automata and Interval Markov Chains. In
this context, we consider the consistency problem that amounts to
deciding whether a given specification admits at least one
implementation. In the context of Interval Probabilistic Timed
Automata (with no timing parameters), we show that this problem is
decidable and propose a constructive algorithm for its resolution. We
show that the existence of timing parameter valuations ensuring consistency
is undecidable in the general context, but still exhibit a syntactic
condition on parameters to ensure decidability. We also propose
procedures that resolve both the consistency and the consistent
reachability problems when the parametric probabilistic zone graph is finite.
\end{abstract}
}

\LongVersion{
	\title{\ourtitle{}\footnote{%
		This is the author version of the manuscript of the same name published in the Journal of Logical and Algebraic Methods in Programming.
		The final version is available at \href{http://www.dx.doi.org/10.1016/j.jlamp.2019.04.007}{10.1016/j.jlamp.2019.04.007}.
		This work is partially supported by the ANR national research program PACS (ANR-14-CE28-0002).
		This work was partially done during Étienne André's \emph{délégation CNRS} at \'Ecole Centrale de Nantes, IRCCyN, CNRS, UMR 6597, France (2015--2016).
	}}
	
	\author{Étienne André$^1$, Benoît Delahaye$^2$ and Paulin Fournier$^2$
	\\
	{\small $^1$ Université Paris 13, LIPN, CNRS, UMR 7030, F-93430, Villetaneuse, France}
	\\
	{\small $^2$ Université de Nantes / LS2N UMR CNRS 6004, Nantes, France}
	}
	
	\date{}
}

\begin{document}

\ShortVersion{
\begin{frontmatter}

\title{\ourtitle{}\tnoteref{thanksANR}%
}
\tnotetext[thanksANR]{This work is partially supported by the ANR national research program PACS (ANR-14-CE28-0002).}

\author{Étienne André\corref{cor1}} %
\cortext[cor1]{This work was partially done during Étienne André's \emph{délégation CNRS} at \'Ecole Centrale de Nantes, IRCCyN, CNRS, UMR 6597, France (2015--2016)}
\ead{first.last@lipn.univ-paris13.fr}
\ead[url]{https://lipn.univ-paris13.fr/~andre/}
\address{Université Paris 13, LIPN, CNRS, UMR 7030, F-93430, Villetaneuse, France} %

\author{Benoît Delahaye, Paulin Fournier\corref{cor2}}
\ead{first.last@univ-nantes.fr}
\ead[url]{http://pagesperso.lina.univ-nantes.fr/~delahaye-b/}
\address{Université de Nantes / LS2N UMR CNRS 6004, Nantes, France} %
\thispagestyle{plain}

\ourabstract{}


\end{frontmatter}

}
\LongVersion{
	\maketitle
	
	\ourabstract{}
}

\ifdefined \VersionWithComments
	\textcolor{red}{\textbf{This is the version with comments. To disable comments, comment out line~3 in the \LaTeX{} source.}}
\fi

\ifdefined\VersionWithComments
	\tableofcontents{}
\fi

\section{Introduction}\label{section:introduction}

\paragraph{Motivation}
Nowadays, automata-based modeling and verification methods are mainly
used in two different ways: for designing digital systems based on
(mostly informal) specifications expressed by the end-users of these
systems or from the knowledge designers have of their environment; and
in order to abstract existing (not necessarily software) systems that
are too complex to comprehend in their entirety. In both cases the
complexity of the systems being designed calls for increasingly
expressive abstraction artifacts such as time and probabilities. Timed
automata, introduced in~\cite{AD94}, are a widely recognized modeling formalism for
reasoning about real-time systems. This modeling formalism, based on
finite control automata equipped with clocks, which are real-valued
variables which increase uniformly at the same rate, has been extended
to the probabilistic framework in~\cite{GJ95,KNSS02}. In this context,
discrete actions are replaced with probabilistic discrete
distributions over discrete actions, allowing to model uncertainties
in the system's behavior. This formalism has been applied to a number
of case studies, \eg{} in~\cite{DBLP:journals/fmsd/KwiatkowskaNPS06}.

Unfortunately, building a system model based either on imprecise
specifications or on imprecise observations often requires to fix
arbitrarily a number of constants in the model, which are then
calibrated by a fastidious comparison of the model behavior and the
expected behavior. This is the case for instance for timing constants
or transition probability values. In order to incorporate these
uncertainties in the model and to develop automatic calibration, more
abstract formalisms have been introduced separately in the timed
setting and in the probabilistic setting.

In the timed setting, {\em
parametric timed automata} (PTAs) introduced by~\cite{AHV93} allow using parameter
variables in the guards of timed transitions in order to account for
the uncertainty on their values.
The reachability emptiness problem, \ie{} the emptiness of the set of valuations for which a given discrete state is reachable, is undecidable for parametric timed automata as shown in~\cite{AHV93}, even for bounded parameters as shown by~\cite{Miller00}, for a single integer-valued parameter as shown by~\cite{BBLS15}, or only when strict inequalities are used as shown by~\cite{Doyen07}.
Decidable subclasses were exhibited (\eg{} \cite{HRSV02,BlT09,JLR15,ALR16ICFEM}).

Parametric probabilistic timed automata were proposed in~\cite{AFS13} to answer the following question:
given a timing parameter valuation, what are other valuations preserving the same minimum and maximum probabilities for reachability properties as the reference valuation?
Parametric probabilistic timed automata were then given a symbolic semantics in~\cite{JK14};
a method has been proposed in that same work to synthesize optimal parameter valuations to maximize or minimize the probability of reaching a discrete location.

In the purely probabilistic setting, Interval Markov Chains (IMCs for short) have
been introduced by~\cite{DBLP:conf/lics/JonssonL91} to take into account imprecision in
the transition probabilities. IMCs extend Markov Chains by allowing to
specify intervals of possible probabilities on transitions instead of
exact values. Methods have then been developed to decide whether
there exist Markov Chains with concrete probability values that match
the intervals specified in a given IMC (see~\cite{DBLP:journals/jlp/DelahayeLLPW12}).

\paragraph{Contribution}
\todo{ÉA: revoir l'intro pour mettre le focus sur l'accessibilité}
In this paper, we propose to combine both abstraction approaches into a single specification theory: Parametric Interval Probabilistic Timed Automata (\PIPROBTAs{} for short).
In this setting, parameters can be used in order to abstract timed constants on transition guards while intervals can be used to abstract imprecise transition probabilities.
Allowing this higher level of freedom allow for incremental design, where one can first give large sets of values for which the system may be defined, and then further refine them.
This refinement will take the form of an instance of a probabilistic interval, or the concrete instance of a timing parameter.

As for IMCs, it is important to be able to decide whether the probability intervals that are specified in a model allow defining consistent probability distributions (\ie{} can be matched in a real-life implementation).
This is called the consistency problem.

First, in the context of Interval Probabilistic Timed Automata with no timing parameters (\IPROBTAs{} for short), we propose an algorithm that solves this problem.

Second, in the parametric setting, since the behavior of the system is
conditioned by the calibration of parameter values, it is necessary to decide whether there exist parameter values that ensure consistency of the resulting model (and synthesize these values when this is possible). 
We show that the existence of such parameter valuations is undecidable in the general context of \PIPROBTAs{}.
Still, we exhibit a sufficient syntactic condition on the use of the parameters to ensure decidability, when parameters are partitioned into lower-bound parameters and upper-bound parameters (in their comparisons with clocks).
In addition, we propose a construction that characterizes, whenever the parametric probabilistic zone graph is finite, the set of parameter values that ensure consistency of the resulting \IPROBTA{}.
We finally address the problem of parametric consistent reachability, \ie{} of synthesizing valuations for which a given state is reachable and the model is consistent.

\begin{example}\label{example:motivating}
	The Root Contention Protocol, used for the election of a leader in the physical layer of the IEEE~1394 standard, consists in first drawing a random number (0 or 1), then waiting for some time according to the result drawn, followed by the sending of a message to the contending neighbor.
	This is repeated by both nodes until one of them receives a message before sending one, at which point the root is appointed.
	This protocol was modeled in~\cite{CS01} using parametric timed automata, in~\cite{KNS03} with probabilistic timed automata, and in~\cite{AFS13}\ea{en fait non ! version conf seulement} using parametric probabilistic timed automata, \ie{} parametric timed automata extended with (non-parametric) probabilistic distributions.
	
	\cref{fig:RCP:node-prob} shows a \PIPROBTA{} model of the node~$i$.
	The wire can be found in \cite{KNS03,AFS13}.
	\cref{fig:RCP:node-prob} features one clock~$x_i$ and four parameters $\rcpFMin{}$, $\rcpFMax{}$, $\rcpSMin{}$ and $\rcpSMax{}$.
	In short, the goal of the protocol is that each node reaches either the child status, or the root status.
	In addition, observe that we use probabilistic interval distributions; they can be seen as an additional design freedom, allowing for incremental design.
	The one going out from ROOT\_IDLE clearly admits no implementation, as no instance of the two intervals $[0.3, 0.4]$ can be such that their sum is equal to~1.
	This probabilistic interval distribution could be either disabled by setting other probabilities to~0 so that location ROOT\_IDLE becomes unreachable; or by tuning the values of the four parameters (or the parameters in the other \PIPROBTAs{} in parallel) so that the guard going out from ROOT\_IDLE becomes unsatisfiable.
	The rest of the this manuscript is dedicated to this problem.
\end{example}

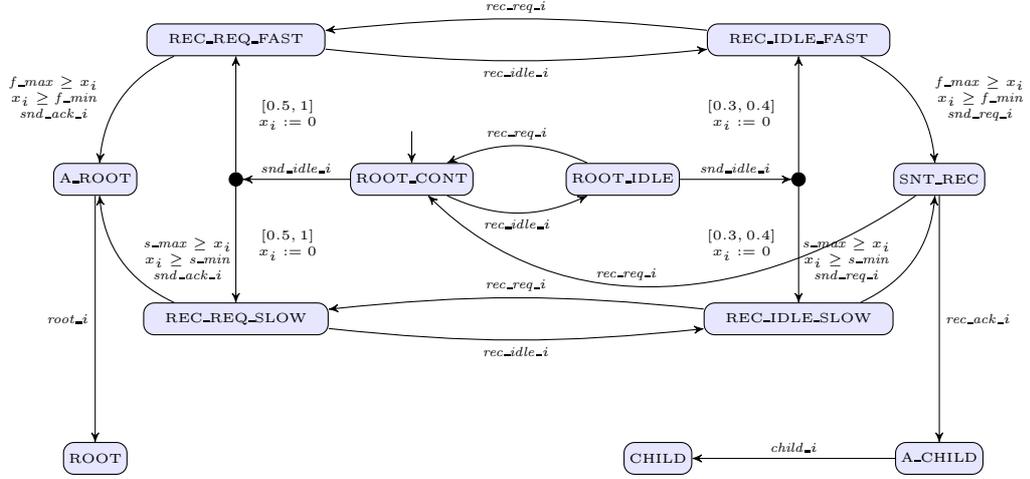
\begin{figure}[ht!]
{

	\centering
	\tiny

	\begin{tikzpicture}[scale = 1.85, auto, ->, >=stealth', thin] %

		\node[location, initial above] at (-0.75, 0) (RC) {\begin{tabular}{@{} c @{} }ROOT\_CONT\end{tabular}}; %
		\node[location] at (+0.75, 0) (RI) {\begin{tabular}{@{} c @{} }ROOT\_IDLE\end{tabular}}; %
		\node[location] at (-2, 1) (RRF) {\begin{tabular}{c}REC\_REQ\_FAST\end{tabular}}; %
		\node[location] at (-3, 0) (AR) {\begin{tabular}{@{} c @{} }A\_ROOT\end{tabular}}; %
		\node[location] at (+2, 1) (RIF) {\begin{tabular}{c}REC\_IDLE\_FAST \end{tabular}}; %
		\node[location] at (+3, 0) (SR) {SNT\_REC};
		\node[location] at (-2, -1) (RRS) {\begin{tabular}{c}REC\_REQ\_SLOW \end{tabular}}; %
		\node[location] at (+2, -1) (RIS) {\begin{tabular}{c}REC\_IDLE\_SLOW \end{tabular}};  %
		\node[location] at (-3, -2) (R) {ROOT};
		\node[location] at (+3, -2) (AC) {\begin{tabular}{@{} c @{} }A\_CHILD\end{tabular}}; %
		\node[location] at (+1, -2) (C) {CHILD};

		\node[point] at (-2, 0) (pRC) {};
		\node[point] at (+2, 0) (pRI) {};

		\path
		(RC)
			edge [bend angle=25, below, bend right] node {$\mathit{rec\_idle\_}i$} (RI)
			edge [above] node {$\mathit{snd\_idle\_}i$} (pRC)
		(pRC)
			edge [right] node {\begin{tabular}{c}$[0.5 , 1]$ \\ $x_i := 0$ \end{tabular}} (RRF)
			edge [right] node {\begin{tabular}{c}$[0.5 , 1]$ \\ $x_i := 0$ \end{tabular}} (RRS)
		(RI)
			edge [bend angle=25, above, bend right] node {$\mathit{rec\_req\_}i$} (RC)
			edge [above] node {$\mathit{snd\_idle\_}i$} (pRI)
		(pRI)
			edge [left] node {\begin{tabular}{c}$[0.3 , 0.4]$ \\ $x_i := 0$ \end{tabular}} (RIF)
			edge [left] node {\begin{tabular}{c}$[0.3 , 0.4]$ \\ $x_i := 0$ \end{tabular}} (RIS)
		(RRF)
			edge [bend angle=6, below, bend right] node {$\mathit{rec\_idle\_}i$} (RIF)
			edge [left, bend right] node {\begin{tabular}{c}$\rcpFMax \geq x_i $\\$ x_i \geq \rcpFMin$ \\ $\mathit{snd\_ack\_}i$ \end{tabular}} (AR)
		(AR)
			edge [left] node {$\mathit{root\_}i$} (R)
		(RIF)
			edge [bend angle=6, above, bend right] node {$\mathit{rec\_req\_}i$} (RRF)
			edge [right, bend left] node {\begin{tabular}{c}$\rcpFMax \geq x_i$\\$ x_i \geq \rcpFMin$ \\ $\mathit{snd\_req\_}i$ \end{tabular}} (SR)
		(SR)
			edge [right] node {$\mathit{rec\_ack\_}i$} (AC)
			edge [above left, bend left, out=+35,in=+135] node {$\mathit{rec\_req\_}i$} (RC)
		(RRS)
			edge [bend angle=6, below, bend right] node {$\mathit{rec\_idle\_}i$} (RIS)
			edge [right, bend left] node {\begin{tabular}{c}$\rcpSMax \geq x_i$\\$ x_i  \geq \rcpSMin$ \\ $\mathit{snd\_ack\_}i$ \end{tabular}} (AR)
		(RIS)
			edge [bend angle=6, above, bend right] node {$\mathit{rec\_req\_}i$} (RRS)
			edge [left, bend right] node {\begin{tabular}{c}$\rcpSMax \geq x_i$\\$ x_i  \geq \rcpSMin$ \\ $\mathit{snd\_req\_}i$ \end{tabular}} (SR)
		(AC)
			edge [above] node {$\mathit{child\_}i$} (C)

		;
		
	\end{tikzpicture}

	}

	\caption{\PIPROBTA{} modeling node~$i$ in the Root Contention Protocol}
	\label{fig:RCP:node-prob}
\end{figure}

\paragraph{Outline}
We start \cref{section:preliminaries} with preliminary
definitions and then introduce the concepts of \IPROBTAs{}
and \PIPROBTAs{}. In \cref{section:Consistency-IPROBTA}, we
study the consistency problem for \IPROBTAs{} and propose a constructive algorithm based on the zone-graph construction that decides whether an
\IPROBTA{} is consistent and produces an
implementation if one exists.
In \cref{section:consistency-PIPROBTA}, we move to the general problem of consistency of \PIPROBTAs{}.
We first show that this problem is undecidable in general and then exhibit a decidable subclass.
We then propose a construction that
characterizes, whenever the parametric probabilistic zone graph is finite, the set of parameter values
ensuring consistency of the resulting \IPROBTA{}.
We also consider the problem of parametric consistent reachability.
Finally, \cref{section:conclusion} concludes the paper.

\section{Preliminaries}
\label{section:preliminaries}
\subsection{Clocks, parameters and constraints}

Let $\grandn$, $\grandz$, $\grandqplus$ and $\grandrplus$ denote the
sets of non-negative integers, integers, non-negative rational numbers
and non-negative real numbers respectively.  Given an arbitrary set
$S$, we write $\dist(S)$ for the set of probabilistic distributions
over~$S$.

Throughout this paper, let $\Clock
= \{ \clock_1, \dots, \clock_\ClockCard \} $ be a set
of \emph{clocks}, \ie{} real-valued variables that evolve at the same
rate, and $\Param = \{ \param_1, \dots, \param_\ParamCard \} $ be a
set of \emph{parameters}, \ie{} unknown constants used in guards.

A clock valuation is a function
$\clockval : \Clock \rightarrow \grandrplus$.
We identify a clock valuation~$\clockval$ with the \emph{point} $(\clockval(\clock_1), \dots, \clockval(\clock_{\ClockCard}))$.
We write $\vec{0}$ for the valuation that assigns $0$ to each clock.
Given $d \in \grandrplus$, $\clockval + d$ denotes the valuation such
that $(\clockval + d)(\clock) = \clockval(\clock) + d$, for all
$\clock \in \Clock$. Given $\resets \subseteq \Clock$, we define
$\reset{\clockval}{\resets}$ as the clock valuation obtained by
resetting the clocks in $\resets$ and keeping the other clocks unchanged.

A parameter {\em valuation} $\pval$ is a function
$\pval : \Param \rightarrow \grandqplus$.
We identify a parameter valuation~$\pval$ with the \emph{point} $(\pval(\param_1), \dots, \pval(\param_{\ParamCard}))$.

In the following, we assume %
${\compOp} \in \{<, \leq, \geq, >\}$.
Let $\lterm$ range over affine terms over $\Clock \cup \Param$, of the form $\sum_{1 \leq i \leq \ClockCard} \alpha_i \clock_i + \sum_{1 \leq j \leq \ParamCard} \beta_j \param_j + d$, with
	$\clock_i \in \Clock$,
	$\param_j \in \Param$,
	and
	$\alpha_i, \beta_j, d \in \grandz$. Similarly, let
$\plterm$ range over parametric affine terms over $\Param$, that is affine terms without clocks ($\alpha_i = 0$ for all $i$).
A \emph{constraint}~$\C$ over $\Clock \cup \Param$ is a conjunction of inequalities of the form $\lterm \compOp 0$ (\ie{} a convex polyhedron).
Given a parameter valuation~$\pval$, $\valuate{\C}{\pval}$ denotes the constraint over~$\Clock$ obtained by replacing each parameter~$\param$ in~$\C$ with~$\pval(\param)$.
Likewise, given a clock valuation~$\clockval$, $\valuate{\valuate{\C}{\pval}}{\clockval}$ denotes the expression obtained by replacing each clock~$\clock$ in~$\valuate{\C}{\pval}$ with~$\clockval(\clock)$.
We say that %
$\pval$ \emph{satisfies}~$\C$,
denoted by $\pval \models \C$,
if the set of clock valuations satisfying~$\valuate{\C}{\pval}$ is nonempty.
Given a parameter valuation $\pval$ and a clock valuation $\clockval$, we denote by $\wv{\clockval}{\pval}$ the valuation over $\Clock\cup\Param$ such that 
for all clocks $\clock$, $\valuate{\clock}{\wv{\clockval}{\pval}}=\valuate{\clock}{\clockval}$
and 
for all parameters $\param$, $\valuate{\param}{\wv{\clockval}{\pval}}=\valuate{\param}{\pval}$.
We use the notation $\wv{\clockval}{\pval} \models \C$ to indicate that $\valuate{\valuate{\C}{\pval}}{\clockval}$ evaluates to true.
We say that $\C$ is \emph{satisfiable} if $\exists \clockval, \pval \text{ \st{}} \wv{\clockval}{\pval} \models \C$.
We define the \emph{time elapsing} of~$\C$, denoted by $\timelapse{\C}$, as the constraint over $\Clock$ and $\Param$ obtained from~$\C$ by delaying all clocks by an arbitrary amount of time.
Given $\resets \subseteq \Clock$, we define the \emph{reset} of~$\C$, written $\reset{\C}{\resets}$, as the constraint obtained from~$\C$ by resetting the clocks in~$\resets$, and keeping the other clocks unchanged.
We denote by $\projectP{\C}$ the projection of~$\C$ onto~$\Param$, \ie{} obtained by eliminating the clock variables (\eg{} using the Fourier-Motzkin algorithm). %

A \emph{guard}~$\guard$ is a constraint over $\Clock \cup \Param$ defined by inequalities of the form
	$\clock \compOp z$, where $\clock\in\Clock$ and $z$ is either a parameter or a constant in~$\grandz$.

A \emph{zone} is a polyhedron over a set of clocks in which all constraints on variables are of the form $\clock\compOp k$ (rectangular constraints) or $\clock_i -\clock_j \compOp k$ (diagonal constraints), where $\clock_i \in \Clock$, $\clock_j \in \Clock$ and $k$ is an integer.
Operations on zones are well-documented (see \eg{} \cite{BY03}).

A \emph{parametric zone} is a convex polyhedron over $\Clock \cup \Param$ in which all constraints on variables are of the form $\clock\compOp \plterm$ (parametric rectangular constraints) or $\clock_i -\clock_j \compOp \plterm$ (parametric diagonal constraints), where $\clock_i \in \Clock$, $\clock_j \in \Clock$ and $\plterm$ is a parametric affine term over $\Param$.
We denote the set of all parametric zones by \PZones{}.

\subsection{Probabilistic timed automata}\label{ss:TPS}

We start by reviewing the definition of timed probabilistic systems, as defined in~\cite{KNSS02}. %
A \emph{timed probabilistic system (TPS)} is a tuple
$\TPS = (\States, \stateinit, \Actions, \StepsIPTA)$ where
$\States$ is a set of \emph{states},
$\stateinit \in \States$ is the \emph{initial state},
$\Actions$ is a finite set of {\em actions}, and
$\StepsIPTA \subseteq \States \times \grandrplus \times \Actions \times \dist(\States)$
is a {\em probabilistic transition relation} that associates a probabilistic distribution over~$\States$ to triples made of a source state in~$\States$, a time in~$\grandrplus$ and an action in~$\Actions$.
Probabilistic timed automata (defined by~\cite{GJ95,KNSS02}) are an extension of classical timed automata (defined in~\cite{AD94}) with discrete probability distributions.

\subsubsection{Syntax}

\begin{definition}\label{def:probTA}
	A Probabilistic Timed Automaton (\PROBTA) $\probta$ is a tuple
        $(\Actions, \Loc, \locinit,\linebreak[4] \Clock, \TransitionsPTA)$, where: %
	\begin{inparaenum}[\itshape i\upshape)]
		\item $\Actions$ is a finite set of actions,
		\item $\Loc$ is a finite set of locations,
		\item $\locinit \in \Loc$ is the initial location,
		\item $\Clock$ is a finite set of clocks,
		\item $\TransitionsPTA$ is a {\em probabilistic edge relation}
		consisting of elements of the form $(\loc,\guard,\action,\DistPTA)$,
		where
		$\loc \in \Loc$, $\guard$  is a zone over the clocks~$\Clock$,
		$\action \in \Actions$, and $\DistPTA \in \dist(2^\Clock \times \Loc)$.
	\end{inparaenum}
\end{definition}

Note that we use no invariant; this is an important condition for the correctness of our techniques.
However, invariants can be eliminated (moved to the guards prior to the transition), following classical techniques defined for (probabilistic) timed automata.\ea{note pour moi pour un jour : give réf}

We use the following conventions for the graphical representation of
probabilistic timed automata: locations are represented by nodes,
within which name %
of the location is written;
probabilistic edges are represented by arcs from locations, labeled
by the associated guard and action, and which split into multiple arcs,
each of which leads to a location and which is labeled by a set of
clocks to be reset to 0 and a probability (probabilistic edges which
correspond to probability 1 are illustrated by a single arc from
location to location).

\begin{example}
	\cref{figure:example:ProbTA} presents an example of a \PROBTA{} with two clocks $x$ and~$y$.
	For example, $\loc_0$ can be exited whenever $y < 2$; then, with probability $0.4$ the target location becomes $\loc_2$, resetting~$x$; or with probability $0.6$ the target location is~$\loc_1$, resetting~$y$.
	The transition from~$\loc_2$ can be explained similarly.
\end{example}

\subsubsection[Semantics of \PROBTAs{}]{Semantics of \PROBTAs{}}

A \PROBTA{} can be interpreted as an infinite TPS.
Due to the continuous nature of clocks, the underlying TPS has uncountably many states, and is uncountably branching.

\begin{definition}[Concrete semantics of a \PROBTA{}]
	Given a \PROBTA{} $\probta = (\Actions, \Loc, \locinit, \Clock,\linebreak[4] \TransitionsPTA)$, where $\ClockCard = |\Clock|$,
	the concrete semantics of $\probta$ is given by the timed probabilistic system $\TPS_\probta = (\States, \stateinit, \Actions, \StepsIPTA)$, with
	\begin{itemize}
		\item $\States = \{ (\loc, \clockval) \in \Loc \times \grandrplus^\ClockCard \}$ %
            , $\stateinit = (\locinit, \vec{0}) $
		\item  $ ((\loc, \clockval), d, \action, \distributionTPS) \in \StepsIPTA$ if both of the following conditions hold:
		\begin{enumerate}
			\item time elapse: $\forall d' \in [0, d], (\loc, \clockval+d') \in \States$, and
			\item edge traversal:
				there exists a probabilistic edge $\pedge = (\loc,\guard,\action, \DistPTA) \in \TransitionsPTA$ such that $\clockval + d \models \guard$ and,
				for each $\loc' \in \Loc$ and $\resets \subseteq \Clock$,
				$\distributionTPS(\loc', \reset{\clockval+d}{\resets}) = \DistPTA(\resets, \loc')$.
		\end{enumerate}
	\end{itemize}
\end{definition}

Note that, due to the fact that we have no invariants, the first condition (time elapse) is always trivially true.

\subsection{Parametric interval probabilistic timed automata}
In this section, we introduce basic definitions for {\em (parametric)
interval probabilistic timed automata}, that extend (parametric)
probabilistic timed automata by providing {\em intervals} for
transition probabilities instead of exact probability values.
In the spirit of (parametric) Interval Markov Chains defined in~\cite{BDSyncop15,DBLP:conf/vmcai/DelahayeLP16}, (parametric)
interval probabilistic timed automata are used for specifying
potentially infinite families (sets) of probabilistic
timed automata---those whose exact probability values match the
specified intervals---with a finite structure of similar form.

\subsubsection{Syntax}\label{sss:PIPTA:syntax}

Given an arbitrary measurable set $S$, we call an {\em interval distribution}
over $S$ a function $\DistIPTA$ that assigns to each element of $S$ an
interval of probabilities $[a,b] \subseteq [0,1]$. Intuitively, an
interval distribution $\DistIPTA$ over $S$ represents the set of all
distributions $\mu \in \dist(S)$ that assign to each element $s \in
S$ a probability $\mu(s)$ such that
$\mu(s) \in \DistIPTA(s)$.
Formally, let $\IntDist(S)$ denote the set of all interval distributions over~$S$;
we define the implementation of an interval distribution as follows.

\begin{definition}[Implementation of an interval distribution]
	Let $S$ be an arbitrary set. Given an interval
	distribution~$\DistIPTA \in \IntDist(S)$,
	$\DistPTA \in \dist(S)$ is
	an \emph{implementation} of~$\DistIPTA$, written
	$\DistPTA \in \DistIPTA$ iff, for all $s \in S$, we have
	$\DistPTA(s) \in \DistIPTA(s)$.
\end{definition}

We now move to the definition of (parametric) interval probabilistic timed automata.

\begin{definition}\label{def:PIPTA}
	A Parametric Interval Probabilistic Timed Automaton\linebreak[4]
        (\PIPROBTA{}) $\piprobta$ is a tuple $(\Actions, \Loc, \locinit, \Clock, \Param, \TransitionsIPTA)$, %
        where: \begin{inparaenum}[\itshape i\upshape)] \item
        $\Actions$ is a finite set of actions, \item $\Loc$ is a
        finite set of locations, \item $\locinit \in \Loc$ is the
        initial location, \item $\Clock$ is a finite set of
        clocks, \item $\Param$ is a finite set of parameters,
		\item $\TransitionsIPTA$ is an {\em interval-valued
		probabilistic edge relation} consisting of elements of
		the form $(\loc,\guard,\action,\DistIPTA)$, where
		$\loc \in \Loc$, $\guard$ is a guard,
		$\action \in \Actions$, and
		$\DistIPTA \in \IntDist(2^\Clock \times \Loc)$ is
		an interval distribution.  \end{inparaenum}
\end{definition}

Given a \PIPROBTA{}~$\piprobta = (\Actions, \Loc, \locinit, \Clock, \Param, \TransitionsIPTA)$ and a parameter valuation~$\pval$,
the \emph{valuation} of $\piprobta$ with~$\pval$, written $\valuate{\piprobta}{\pval}$,
is an Interval Probabilistic Timed Automaton (\IPROBTA{}) $\iprobta = (\Actions, \Loc, \locinit, \Clock, \TransitionsIPTA')$,
where $\TransitionsIPTA'$ is obtained by replacing within~$\TransitionsIPTA$ any occurrence of a parameter~$\param$ with $\pval(\param)$ and removing all transitions $(\loc,\guard,\action,\DistIPTA)$ such that $\valuate{\guard}{\pval} \equiv \bot$ (technically, this latter part is not strictly speaking necessary, but it syntactically reduces the model a bit)\ea{rajouté ça pour faire plaisir à R2}.

Remark that \IPROBTAs{} are very similar to \PROBTAs{}: the only
difference is that probabilistic edges are labeled with intervals
instead of exact probability values.

In our graphical representations, when the interval associated with a
distribution is reduced to a point (\eg{} $[0.5, 0.5]$, we simply
represent it using its punctual value (\ie{} $0.5$).  Also, when a
distribution is made of a single target location with probability~1,
we simply omit the distribution.

\begin{figure}

	\begin{subfigure}[b]{\textwidth}
	{\centering
	
	\newcommand{\hscale}{1.5}
	\newcommand{\vscale}{1}
	\begin{tikzpicture}[node distance=2cm, auto, ->, >=stealth']
		\node[location0, initial] (l0) at (0,0) {$\loc_0$};
		\node [probchoice] (l0choice) at (1*\hscale, 0*\vscale) {};

		\node[location1] (l1) at (2*\hscale, -1*\vscale) {$\loc_1$};
		\node[location2] (l2) at (2*\hscale, +1*\vscale) {$\loc_2$};

		\node[location5] (l5) at (6*\hscale, +1*\vscale) {$\loc_5$};

                \node[location2] (l6) at (4*\hscale, 2*\vscale) {$\loc_2'$};
		
		\node [probchoice] (l2choice) at (3*\hscale, +1*\vscale) {};

		\node [probchoice] (l6choice) at (5*\hscale, +2*\vscale) {};

		\path
			(l0) edge node[below]{\begin{tabular}{c}\footnotesize $y < 2$ \\ $a$\end{tabular}} (l0choice)
		
			(l0choice) edge[probedge] node[below]{\footnotesize \begin{tabular}{c}$0.6$ \\ $y := 0$\end{tabular}} (l1)
			(l0choice) edge[probedge] node[xshift=8, yshift=-8]{\footnotesize \begin{tabular}{c}$0.4$ \\ $x := 0$\end{tabular}} (l2)
			
			(l2) edge node[below,yshift=-3]{\footnotesize \begin{tabular}{c}$x = 1 \land y \leq 2$ \\ $c$ \end{tabular}} (l2choice)

			(l2choice) edge[probedge, out=30, in=60] node[above, xshift=-10, yshift=-3]{\footnotesize \begin{tabular}{c}$0.1$ \\ $x := 0$\end{tabular}} (l2)
			(l2choice) edge[probedge] node[above, xshift=-3, yshift=-2]{\footnotesize \begin{tabular}{c}$0.1$ \\ $x := 0$\end{tabular}} (l6)
			(l2choice) edge[probedge] node[below]{\footnotesize $0.8$} (l5)

			(l6) edge node[below,yshift=0]{\footnotesize \begin{tabular}{c}$x = 1 \land y \leq 2$ \\ $c$ \end{tabular}} (l6choice)

			(l6choice) edge[probedge, out=30, in=60] node[above, yshift=-3, xshift=-8]{\footnotesize \begin{tabular}{c}$0.1$ \\ $x := 0$\end{tabular}} (l6)
			(l6choice) edge[probedge] node{\footnotesize $0.9$} (l5)

                ;
		
		\AngleA{l0choice}{l1}{l2};

		\node[draw=none] (l2fake) at (3.5*\hscale, 2*\vscale){};
		\AngleA{l2choice}{l2fake}{l5};

		\node[draw=none] (l6fake) at (5.5*\hscale, 3*\vscale){};
		\AngleA{l6choice}{l6fake}{l5};

\end{tikzpicture}
	
	}
		\caption{A \PROBTA{}}
		\label{figure:example:ProbTA}
	\end{subfigure}

	\bigskip
	\bigskip

	\begin{subfigure}[b]{\textwidth}
	{\centering
	
	\newcommand{\hscale}{1.7}
	\newcommand{\vscale}{1.1}
	\begin{tikzpicture}[node distance=2cm, auto, ->, >=stealth']
		\node[location0, initial] (l0) at (0,0) {$\loc_0$};
		\node [probchoice] (l0choice) at (1*\hscale, 0*\vscale) {};

		\node[location1] (l1) at (2*\hscale, -1*\vscale) {$\loc_1$};
		\node[location2] (l2) at (2*\hscale, +1*\vscale) {$\loc_2$};

		\node [probchoice] (l1choice) at (3*\hscale, -1*\vscale) {};

		\node[location5] (l5) at (4*\hscale, 1*\vscale) {$\loc_5$};
		
		\node [probchoice] (l2choice) at (3*\hscale, +1*\vscale) {};

		\node[location3] (l3) at (5*\hscale, -2*\vscale) {$\loc_3$};
		\node[location4] (l4) at (5*\hscale, 0*\vscale) {$\loc_4$};
		
		\path
			(l0) edge node[below]{\footnotesize\begin{tabular}{c}$y < 2$ \\ $a$ \end{tabular}} (l0choice)
		
			(l0choice) edge[probedge] node[below]{\footnotesize\begin{tabular}{c}$[0,1]$ \\ $y := 0$\end{tabular}} (l1)
			(l0choice) edge[probedge] node[xshift=8, yshift=-8]{\footnotesize\begin{tabular}{c}$[0,0.5]$ \\ $x := 0$\end{tabular}} (l2)
			
			(l1) edge node[below]{\footnotesize\begin{tabular}{c}$2 \leq x \leq \param$ \\ $b$ \end{tabular}} (l1choice)

			(l1choice) edge[probedge] node[below]{\footnotesize\begin{tabular}{c}$[0,0.2]$ \\ $y := 0$\end{tabular}} (l3)
			(l1choice) edge[probedge] node[yshift=-5]{\footnotesize\begin{tabular}{c}$[0,0.3]$ \\ $x,y := 0$\end{tabular}} (l4)

			(l2) edge node[below,yshift=-3]{\footnotesize\begin{tabular}{c}$x = 1 \land y \leq 2$ \\ $c$\end{tabular}} (l2choice)

			(l2choice) edge[probedge, out=30, in=60] node[above, xshift=-10, yshift=-3]{\footnotesize\begin{tabular}{c}$[0,0.2]$ \\ $x := 0$\end{tabular}} (l2)
			(l2choice) edge[probedge] node{\footnotesize$[0.8,1]$} (l5)
			
			(l3) edge[bend left] node[left,xshift=8]{\footnotesize\begin{tabular}{c}$x = 5$\\ $d$ \\ $x,y := 0$\end{tabular}} (l4)
			(l4) edge[bend left] node[right]{\footnotesize\begin{tabular}{c}$2 \leq x \leq \param$\\ $e$ \\ $x := 0$\end{tabular}} (l3)
		;
		
		\AngleA{l0choice}{l1}{l2};
		
		\AngleA{l1choice}{l4}{l3};

		\node[draw=none] (l2fake) at (3.5*\hscale, 2*\vscale){};
		\AngleA{l2choice}{l2fake}{l5};
	\end{tikzpicture}
	
	}
		\caption{A \PIPROBTA{}}
		\label{figure:example:PIPROBTA}
	\end{subfigure}

	\caption{Examples}
\end{figure}
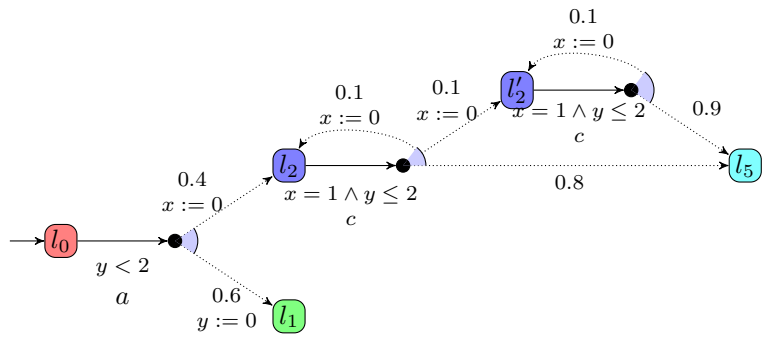
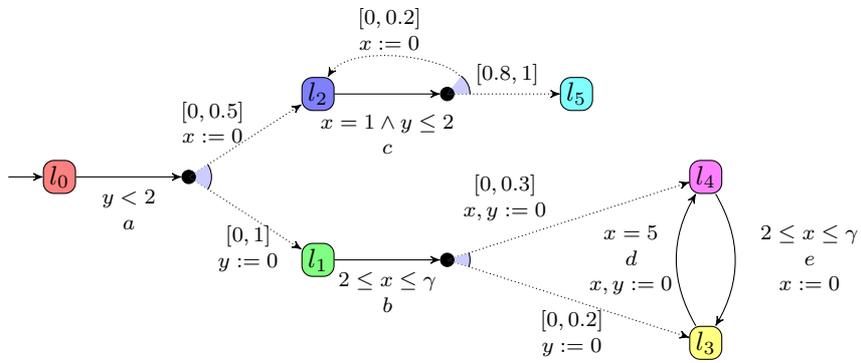

Once a parameter valuation is fixed, the resulting \IPROBTA{}
represents a potentially infinite set of \PROBTAs{}. In order to relate
a given \IPROBTA{} with the \PROBTAs{} it represents, we use the notion
of {\em implementation} defined hereafter.
This notion is similar to the one defined in the context of (parametric) Interval Markov Chains in~\cite{BDSyncop15,DBLP:conf/vmcai/DelahayeLP16}. Remark that a \PROBTA{} implementing
an \IPROBTA{} needs to conserve the exact same clocks, guards and resets.

\begin{definition}[Implementation of an \IPROBTA{}]\label{definition:IPROBTA:implementation}
	Let $\probta = (\Actions, \Loc, \locinit, \Clock, \TransitionsPTA)$ be
	a \PROBTA{} and $\iprobta =
	(\Actions, \Loc', \locinit', \Clock, \TransitionsIPTA)$ be
	an \IPROBTA{}.
	
	We say that $\probta$ is an implementation of $\iprobta$, written $\probta \models \iprobta$, iff there exists a relation $ \RelSimPTA \subseteq \Loc \times \Loc'$, called an {\em implementation relation} \st{}
	$(\loc_0, \loc_0') \in \RelSimPTA$ and, whenever $(\loc, \loc') \in \RelSimPTA$, we have
	\begin{itemize}
		\item $\forall (\loc,\guard, \action, \DistPTA) \in \TransitionsPTA
	, \exists (\loc', \guard, \action, \DistIPTA) \in \TransitionsIPTA$ \st{} $\DistPTA \preceq_{\RelSimPTA} \DistIPTA$, and
		\item $\forall (\loc', \guard', \action, \DistIPTA) \in \TransitionsIPTA, \exists (\loc,\guard', \action, \DistPTA) \in \TransitionsPTA$ \st{} 
	$\DistPTA \preceq_{\RelSimPTA} \DistIPTA$,
	\end{itemize}
	where $ \DistPTA \preceq_{\RelSimPTA} \DistIPTA$ iff
	$\exists \correspondance \in \dist(\Loc \times \Loc')$ \st{}
	\begin{itemize}
		\item $\forall (\resets, \loc) \in 2^\Clock \times \Loc, \DistPTA(\resets, \loc) > 0 \Rightarrow \sum_{\loc' \in \Loc'}(\correspondance(\loc, \loc'))=1$,
		\item $\forall (\resets', \loc') \in 2^\Clock \times \Loc' , \sum_{\loc \in \Loc}( \DistPTA(\resets', \loc) \cdot \correspondance(\loc, \loc')) \in \DistIPTA(\resets', \loc')$, and
		\item $\correspondance(\loc, \loc') > 0 \Rightarrow (\loc, \loc') \in \RelSimPTA$.
	\end{itemize}
\end{definition}

In the above definition, the relation $\RelSimPTA$ encodes the pairs
of states $(\loc,\loc') \in \Loc \times \Loc'$ where $\loc$ is an {\em
implementation} of $\loc'$. On the other hand, the relation
$\preceq_{\RelSimPTA}$ is a lifting of the relation $\RelSimPTA$ to
distributions over locations (also called a {\em coupling}), and
therefore represents {\em compatible} distributions \wrt
$\RelSimPTA$. This notion of satisfaction has been adapted from the
notion of ``weak weak'' satisfaction in the context of Abstract Probabilistic
Automata, for which several notions of satisfaction exist. We have
chosen this particular notion because it is the most permissive among
those presented in~\cite{APAJournal}. For a detailed discussion on
this topic, we refer the interested reader to~\cite{APAJournal}.

Given an \IPROBTA{}, deciding whether the family it represents is
nonempty is a nontrivial problem. Indeed, the interval distributions
used throughout its structure could represent contradictory
constraints on the transition probabilities, therefore preventing
any \PROBTA{} from implementing it.

In the following, we say that a \PROBTA{} $\probta$ has {\em the same
  structure} as an \IPROBTA{} $\iprobta$ if and only if the underlying
directed graph of $\probta$ is a subgraph (up to renaming and removal
of unreachable states) of the underlying directed graph of
$\iprobta$. This notion trivially extends to \PIPROBTA{} and other
models we use in the rest of the paper.

\begin{definition}[Consistency of an \IPROBTA{}]\label{definition:consistency:IPROBTA}
	An \IPROBTA{} is consistent if it admits at least one implementation.
\end{definition}

\begin{example}\label{example:two}
	Consider the \PIPROBTA{}~$\piprobta$ given in \cref{figure:example:PIPROBTA}, and containing a single parameter~$\param$.
		Let $\pval_1$ be the parameter valuation such that $\pval_1(\param) = 1$.
	In the \IPROBTA{} $\valuate{\piprobta}{\pval_1}$, the transition outgoing from~$\loc_1$ can never be taken, as its guard becomes $2 \leq x \leq 1$, which is unsatisfiable.
	Then, it is clear that the \PROBTA{}~$\probta$ given in \cref{figure:example:ProbTA} is an implementation of~$\valuate{\piprobta}{\pval_1}$. We emphasize the fact that location $\loc_2$ from $\valuate{\piprobta}{\pval_1}$ has been ``unfolded'' in $\probta$, yielding two locations $\loc_2$ and $\loc_2'$. As a consequence, the underlying structure of $\probta$ is not identical to the one of $\valuate{\piprobta}{\pval_1}$. Nevertheless, both locations $\loc_2$ and $\loc_2'$ of $\probta$ obviously satisfy the original $\loc_2$ from $\valuate{\piprobta}{\pval_1}$, which allows $\probta$ to satisfy $\valuate{\piprobta}{\pval_1}$ despite their distinct structures.
	As a consequence, $\valuate{\piprobta}{\pval_1}$ is a consistent \IPROBTA{}.
\end{example}

An important problem is therefore to decide whether a given \IPROBTA{} is consistent, which we address in the next section.

\section{The consistency problem for \IPROBTAs{}}\label{section:Consistency-IPROBTA}

In this section, we address the problem of deciding whether a
given \IPROBTA{} is consistent. Unlike in the context of IMCs, where it
is proven that a given IMC is consistent iff it admits an
implementation with the same structure, a given \IPROBTA{} can be
consistent but still not admit any implementation that respects its
structure. 
Indeed, the structure of implementations depends on the structure of the zone graph rather than on the structure of the \IPROBTA{} itself which can be different.
Algorithms such as those proposed for deciding consistency of (p)IMCs
in~\cite{DBLP:conf/vmcai/DelahayeLP16} therefore cannot be directly adapted to the \IPROBTAs{} setting
as they are dependent on this property.

Fortunately, the operational semantics of \IPROBTAs{} can be expressed
in terms of Interval Markov Decision Processes (\IMDPs{}), which are
similar to IMCs and satisfy the same structural properties regarding
consistency.
We therefore propose an algorithm for deciding consistency of
\IPROBTAs{} based on the consistency of their symbolic \IMDP{}
semantics. An alternative solution would be to ``normalize''
\IPROBTAs{} into special \IPROBTAs{} where all edges can fire, via a
region construction. For the sake of simplicity, we only explore the
first solution.
We start with preliminary definitions on \IMDPs{}, then formally define the symbolic semantics of \IPROBTAs{} and finally propose an algorithm for deciding whether a given \IPROBTA{} is consistent.

\subsection{Preliminary definitions}

An \IMDP{} is a tuple $(\States, \stateinit, \Actions, \TransitionsMDP)$
where $\States$ is a set of states, $\stateinit \in \States$ is the
initial state, $\Actions$ is a finite set of actions and
$\TransitionsMDP \subseteq \States \times \Actions \times \IntDist(\States)$
is a probabilistic (interval) transition relation.

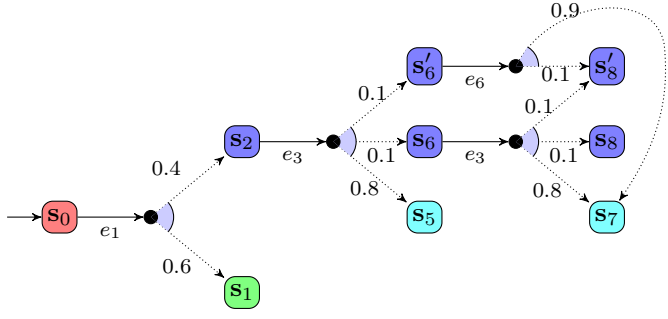
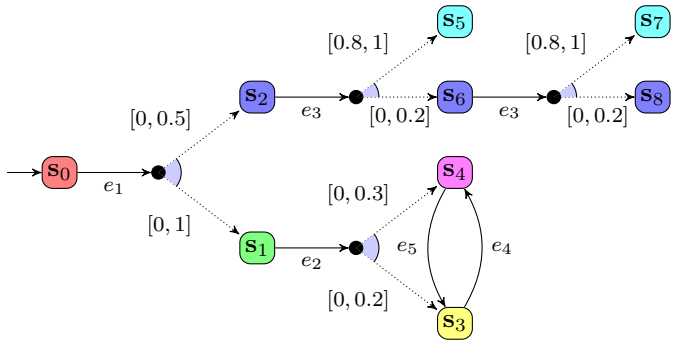
\begin{figure}

	\begin{subfigure}[b]{\textwidth}
	{
	
	\newcommand{\hscale}{1.2}
	\newcommand{\vscale}{1}

	\begin{tikzpicture}[node distance=2cm, auto, ->, >=stealth']
		\node[location0, initial] (l0) at (0,0) {$\symbstate_0$};
		\node [probchoice] (l0choice) at (1*\hscale, 0*\vscale) {};

		\node[location1] (l1) at (2*\hscale, -1*\vscale) {$\symbstate_1$};
		\node[location2] (l2) at (2*\hscale, +1*\vscale) {$\symbstate_2$};

		\node[location5] (l5) at (4*\hscale, 0*\vscale) {$\symbstate_5$};
		\node[location2] (l2b) at (4*\hscale, 1*\vscale) {$\symbstate_6$};
		\node[location2] (l6) at (4*\hscale, 2*\vscale) {$\symbstate_6'$};
		
		\node [probchoice] (l2choice) at (3*\hscale, +1*\vscale) {};

		\node[location5] (l5b) at (6*\hscale, 0*\vscale) {$\symbstate_7$};
		\node[location2] (l6b) at (6*\hscale, 2*\vscale) {$\symbstate_8'$};
		\node[location2] (l2c) at (6*\hscale, 1*\vscale) {$\symbstate_8$};
		
		\node [probchoice] (l2bchoice) at (5*\hscale, +1*\vscale) {};
		\node [probchoice] (l6choice) at (5*\hscale, +2*\vscale) {};

                \coordinate[shift={(5mm,2mm)}] (n) at (l6b.north east);
                \coordinate[shift={(10mm,0mm)}] (m) at (l2c.east);

		\path
			(l0) edge node[below]{\footnotesize$\pedge_1$} (l0choice) %
		
			(l0choice) edge[probedge] node[below, xshift=-5]{\footnotesize$0.6$} (l1)
			(l0choice) edge[probedge] node{\footnotesize$0.4$} (l2)
			
			(l2) edge node[below]{\footnotesize$\pedge_3$} (l2choice) %

			(l2choice) edge[probedge] node[below, xshift=3, yshift=1]{\footnotesize$0.1$} (l2b)
			(l2choice) edge[probedge] node[below, xshift=-3, yshift=1]{\footnotesize$0.8$} (l5)
			(l2choice) edge[probedge] node[above]{\footnotesize$0.1$} (l6)
			
			(l2b) edge node[below]{\footnotesize$\pedge_3$} (l2bchoice) %

			(l2bchoice) edge[probedge] node[below, xshift=3, yshift=1]{\footnotesize$0.1$} (l2c)
			(l2bchoice) edge[probedge] node[below, xshift=-3]{\footnotesize$0.8$} (l5b)
			(l2bchoice) edge[probedge] node[xshift=3, yshift=-5]{\footnotesize$0.1$} (l6b)

			(l6) edge node[below]{\footnotesize$\pedge_6$} (l6choice) %

			(l6choice) edge[probedge] node[below, yshift=3]{\footnotesize$0.1$} (l6b)

		;

		\draw[->, densely dotted, rounded corners=5pt] (l6choice) to[out = 60] node[above, yshift=-5, xshift=-8]{\footnotesize$0.9$} (n) to[bend left =20]  (l5b);
		
		\AngleA{l0choice}{l1}{l2};
		
		\AngleA{l2choice}{l6}{l5};
		\AngleA{l2bchoice}{l6b}{l5b};

		\node[draw=none] (l6fake) at (5.5*\hscale, 3*\vscale){};
		\AngleA{l6choice}{l6fake}{l6b};
                
	\end{tikzpicture}
	
	}
		\caption{An example of an \MDP{}}
		\label{figure:example:MDP}
	\end{subfigure}

	\bigskip

	\begin{subfigure}[b]{\textwidth}
	{%
	
	\newcommand{\hscale}{1.3}
	\newcommand{\vscale}{1}

	\begin{tikzpicture}[node distance=2cm, auto, ->, >=stealth']
		\node[location0, initial] (l0) at (0,0) {$\symbstate_0$};
		\node [probchoice] (l0choice) at (1*\hscale, 0*\vscale) {};

		\node[location1] (l1) at (2*\hscale, -1*\vscale) {$\symbstate_1$};
		\node[location2] (l2) at (2*\hscale, +1*\vscale) {$\symbstate_2$};

		\node [probchoice] (l1choice) at (3*\hscale, -1*\vscale) {};

		\node[location5] (l5) at (4*\hscale, 2*\vscale) {$\symbstate_5$};
		\node[location2] (l2b) at (4*\hscale, 1*\vscale) {$\symbstate_6$};
		
		\node [probchoice] (l2choice) at (3*\hscale, +1*\vscale) {};

		\node[location5] (l5b) at (6*\hscale, 2*\vscale) {$\symbstate_7$};
		\node[location2] (l2c) at (6*\hscale, 1*\vscale) {$\symbstate_8$};
		
		\node [probchoice] (l2bchoice) at (5*\hscale, +1*\vscale) {};

		\node[location3] (l3) at (4*\hscale, -2*\vscale) {$\symbstate_3$};
		\node[location4] (l4) at (4*\hscale, 0*\vscale) {$\symbstate_4$};
		
		\path
			(l0) edge node[below]{\footnotesize$\pedge_1$} (l0choice) %
		
			(l0choice) edge[probedge] node[below left]{\footnotesize$[0,1]$} (l1)
			(l0choice) edge[probedge] node{\footnotesize$[0,0.5]$} (l2)
			
			(l1) edge node[below]{\footnotesize$\pedge_2$} (l1choice) %

			(l1choice) edge[probedge] node[below left]{\footnotesize$[0,0.2]$} (l3)
			(l1choice) edge[probedge] node{\footnotesize$[0,0.3]$} (l4)

			(l2) edge node[below]{\footnotesize$\pedge_3$} (l2choice) %

			(l2choice) edge[probedge] node[below]{\footnotesize$[0,0.2]$} (l2b)
			(l2choice) edge[probedge] node{\footnotesize$[0.8,1]$} (l5)
			
			(l2b) edge node[below]{\footnotesize$\pedge_3$} (l2bchoice) %

			(l3) edge[bend right] node[right]{\footnotesize$\pedge_4$} (l4) %
			(l4) edge[bend right] node[left]{\footnotesize$\pedge_5$} (l3) %

			(l2bchoice) edge[probedge] node[below]{\footnotesize$[0,0.2]$} (l2c)
			(l2bchoice) edge[probedge] node{\footnotesize$[0.8,1]$} (l5b)
			
		;
		
		\AngleA{l0choice}{l1}{l2};
		
		\AngleA{l1choice}{l4}{l3};
		
		\AngleA{l2choice}{l2b}{l5};
		\AngleA{l2bchoice}{l2c}{l5b};

	\end{tikzpicture}
	
	\caption{An example of an \IMDP{}}
	\label{figure:example:IMDP}
	}
	\end{subfigure}

	\caption{Examples}
\end{figure}

\begin{example}
	\cref{figure:example:IMDP} depicts an example of an \IMDP{}.
	Just as for \IPROBTAs{}, when the interval associated with a distribution is reduced to a point (which is not the case here), we simply represent it using its punctual value.
	When a distribution is made of a single target location with probability~1, we simply omit the distribution (\eg{} from $\symbstate_3$ to~$\symbstate_4$).
\end{example}

\begin{definition}[\MDP]
	An \MDP{} is an \IMDP{} such that for each $(s,a,\DistIMDP) \in \TransitionsMDP$,
	and for all $s' \in S$, we have $\DistIMDP(s') = [m,m]$ is a singleton.
	In addition, for each $(s,a,\DistIMDP) \in \TransitionsMDP$, we have
	\[\sum_{\state' \in S} \DistIMDP(s') = 1\text{.}\]
\end{definition}

\begin{example}
	\cref{figure:example:MDP} depicts an example of an \MDP{}.
\end{example}

\begin{definition}[Implementation of an \IMDP{}]
\label{def:imdp-impl}
	Let $\imdp = (\States, \stateinit, \Actions, \TransitionsMDP)$ be an \IMDP{}.
	Let $\mdp = (\States', \stateinit', \Actions, \TransitionsMDP')$ be an \MDP{}.
	We say that $\mdp$ is an implementation of~$\imdp{}$, written $\mdp \models \imdp$, if 
	$\exists \RelSimMDP \subseteq \States' \times \States$ \st{}
	$(\state_0', \state_0) \in \RelSimMDP$ and
	$(\state', \state) \in \RelSimMDP$ if
	\begin{itemize}
		\item $\forall (\state', \action , \DistMDP) \in \TransitionsMDP', \exists (\state, \action, \DistIMDP) \in \TransitionsMDP$ \st{} $\DistMDP \preceq_{\RelSimMDP} \DistIMDP$, and
		\item $\forall (\state, \action , \DistIMDP) \in \TransitionsMDP, \exists (\state', \action, \DistMDP) \in \TransitionsMDP'$ \st{} $\DistMDP \preceq_{\RelSimMDP} \DistIMDP$,
	\end{itemize}

	where $\DistMDP \preceq_{\RelSimMDP} \DistIMDP$ iff
	$\exists \correspondance \in \dist(\States' \times \States)$ \st{}
	\begin{itemize}
		\item $\forall \state' \in \States', \DistMDP(\state') > 0 \Rightarrow \sum_{\state \in \States}(\correspondance(\state', \state))=1$,
		\item $\forall \state \in \States, \sum_{\state' \in \States'}( \DistMDP(\state') \cdot \correspondance(\state', \state)) \in \DistIMDP(\state)$, and
		\item $\correspondance(\state', \state) > 0 \Rightarrow (\state', \state) \in \RelSimMDP$.
	\end{itemize}

\end{definition}

As for \IPROBTAs{}, we say that an \IMDP{} is \emph{consistent} iff it admits at least one implementation.

\begin{example}\label{example5}
	The \IMDP{} given in \cref{figure:example:IMDP} admits no implementation:
		indeed, on the (single) transition labeled with~$\pedge_2$, no valuation of the two intervals $[0,0.3]$ and $[0,0.2]$ is such that the sum of both valuations is equal to~1.
		Nevertheless, it could be that the \IMDP{} is still consistent if one assigns a 0-probability on the transition from $\state_0$ to~$\state_1$.
		However, %
		although this would be compatible with the interval ($0 \in [0,1]$), the second interval (to~$\state_2$) does not accept a 1-probability since its probability must be within $[0, 0.5]$.
\end{example}

As said above \IMDPs{} satisfy the same structural property as IMCs
concerning implementations: they are consistent iff they admit at
least one implementation that respects their structure. This result is
formalized in the following lemma.

\begin{lemma}[structure of an implementation]\label{lemma:structure}
	An IMDP~$\imdp$ is consistent iff there exists an MDP~$\mdp$ with the same structure \st{} $\mdp \models \imdp$.
\end{lemma}
\begin{proof}

Let $\imdp = (\States, \stateinit, \Actions, \TransitionsMDP)$ be an \IMDP{}.

One direction of this result is trivial: if there exists an \MDP{} $\mdp$
with the same structure as $\imdp$ \st{} $\mdp \models \imdp$, then
$\imdp$ is clearly consistent.

The reverse implication is more involved.
Assume that $\imdp$ is consistent, \ie{} there exists an \MDP{} $\mdp = (\States', \stateinit', \Actions, \TransitionsMDP')$, with no assumption
on its structure, such that $\mdp \models \imdp$.
We then have to build an \MDP{}
$\mdp^{*} = (\States, \stateinit, \Actions, \TransitionsMDP^{*})$
such that $\mdp^{*} \models \imdp$.
Observe that $\States$ and $\stateinit$ must be identical to that of~$\imdp$ because they have the same structure.

Let $\RelSimMDP$ be the relation witnessing that
$\mdp \models \imdp$ and let $f : \States \rightarrow \States' \cup \{\bot\}$ be
a function that associates to all states in $\imdp$ one of the states
from $\mdp$ that contributes to its implementation, if there is any, and $\bot$ otherwise. Formally,
for all $\state \in \States$, if $f(\state) \ne \bot$ then
$(f(\state),\state) \in \RelSimMDP$, and whenever there exists
$\state' \in \States'$ such that $(\state', \state) \in \RelSimMDP$,
we have $f(\state) \ne \bot$.

The transition relation $\TransitionsMDP^{*}$ of $\mdp^{*}$ is
constructed as follows: For each state $\state$ that is implemented,
\ie{} such that $f(\state) \ne \bot$, and probabilistic interval
transition $(\state,\action,\DistIMDP) \in \TransitionsMDP$ in
$\imdp$, we build a corresponding transition
$(\state,\action,\DistMDP^{\DistIMDP})$ in $\mdp^{*}$ from the
transitions in $\mdp$ that implement $(\state,\action,\DistIMDP)$. In
other words, we pick one of the states that satisfy $\state$ (using
function $f$) and mimic its outgoing transitions in $\mdp^{*}$. All
the other states that satisfy $\state$ are simply removed.  States
that are not implemented do not serve for consistency and are
therefore not considered.

Formally, let $(\state_1,\action,\DistIMDP) \in \TransitionsMDP$ be a probabilistic
interval transition in $\imdp$. From
\cref{def:imdp-impl}, we know that there exists
$(f(\state_1), \action, \DistMDP) \in \TransitionsMDP'$ \st{}
$\DistMDP \preceq_{\RelSimMDP} \DistIMDP$.
According to the definition, there exists at least one function $\correspondance$ that witnesses $\DistMDP \preceq_{\RelSimMDP} \DistIMDP$. In the following we pick one such function and name it $\correspondance_{(\DistMDP,\DistIMDP)}$.
The distribution $\DistMDP^{\DistIMDP}$ is then constructed as follows:
for all $\state_2 \in \States$, let $\DistMDP^{\DistIMDP}(\state_2)
= \sum_{\state' \in \States'} \DistMDP(\state') \cdot \correspondance_{(\DistMDP,\DistIMDP)}(\state',\state_2)$.

By definition of $\correspondance_{(\DistMDP,\DistIMDP)}$, observe that
$\DistMDP^{\DistIMDP}(\state_2) \in \DistIMDP(\state_2)$ for all
$\state_2 \in \States$ and that, whenever
$\DistMDP^{\DistIMDP}(\state_2) >0$, $f(\state_2)\ne \bot$.

Clearly, $\mdp^{*}$ is therefore an implementation of $\imdp$, with
witnessing relation $\RelSimMDP^{*}$ defined as the identity relation on the set of
states $\state \in \States$ such that $f(\state)\ne \bot$.
\end{proof}

\subsection{A symbolic semantics for \IPROBTAs{}}

We equip \IPROBTAs{} with a symbolic semantics, defined below.
Basically, it is inline with the symbolic semantics defined for timed automata in the form of a zone graph, %
with the addition of probabilistic intervals on the edges; as a consequence, the semantics becomes not an LTS, but an \IMDP{}.

\begin{definition}[Symbolic semantics of an \IPROBTA{}]\label{definition:IPTA:symbolic-semantics}
	Given an \IPROBTA{} $\iprobta = (\Actions, \Loc, \locinit, \Clock, \TransitionsIPTA)$,
	the symbolic semantics of $\iprobta$ is given by the \IMDP{} $(\SymbStates, \symbstateinit, \TransitionsIPTA, \TransitionsMDP)$, with
	\begin{itemize}
		\item $\SymbStates = \{ (\loc, \C) \in \Loc \times \PZones \}$, %
			$\symbstateinit = (\locinit, \timelapse{(\bigwedge_{1 \leq i\leq\ClockCard}\clock_i=0)} )$, where $\loc$ is the location and $\C$ the associated zone, %
		\item $((\loc, \C), \pedge, \DistIMDP) \in \TransitionsMDP$ if $\pedge = (\loc,\guard,\action,\DistIPTA) \in \TransitionsIPTA$ and for all $\loc' \in \Loc$, for all $\resets \subseteq \Clock$ such that $\DistIPTA(\resets, \loc') > 0$, 
		$\C' = \timelapse{\big(\reset{\C \land \guard}{\resets}\big )} $, %
		and $\DistIMDP((\loc', \C')) = \DistIPTA(\resets, \loc')$.
	\end{itemize}
\end{definition}

Given a symbolic state $\symbstate = (\loc, \C)$, we denote by $\symbstate.\loc$ and $\symbstate.\C$ its location and its associated zone (symbolic constraint), respectively.

Observe that, whenever an \IPROBTA{} has no probabilistic choice, then the \IMDP{} becomes a labeled transition system, and the symbolic semantics matches that of timed automata given in the form of a zone graph (see \eg{}~\cite{BY03}).
It is well-known that the zone graph of a timed automaton can have an infinite number of states; however, applying the classical $k$-extrapolation (that basically splits zones between a part where the clock constraints are smaller or equal to~$k$ and a part where constraints are larger than~$k$, where $k$ is the largest integer-constant in the timed automaton) yields finiteness (see, \eg{} \cite{BBLP06}).
In the following, we apply the classical $k$-extrapolation to the symbolic constraints of the semantics of an \IPROBTA{}~$\iprobta$, and therefore the number of states in the \IMDP{} described in \cref{definition:IPTA:symbolic-semantics} is finite.
We refer to the symbolic semantics of~$\iprobta$ as the \emph{probabilistic zone graph} of~$\iprobta$.

Remark that the probabilistic zone graph is defined for \IPROBTAs{} in the form of an \IMDP{}; a \PROBTA{} can be understood as an \IPROBTA{}, and its associated zone graph becomes an \MDP{}.

\begin{table}
\centering
\footnotesize

	\begin{tabular}{| c | c | c |}
		\hline
		\cellHeader{State} & \cellHeader{Location} & \cellHeader{$\C$} \\
		\hline
		$\symbstate_0$ & $\loc_0$ & $x = y \land x \geq 0$ \\
		\hline
		$\symbstate_1$ & $\loc_1$ & $0 \leq x - y < 2 \land y \geq 0 $ \\
		\hline
		$\symbstate_2$ & $\loc_2$ & $0 \leq y - x < 2 \land x \geq 0 $ \\
		\hline
		$\symbstate_5$ & $\loc_5$ & $0 \leq y - x \leq 1 \land x \geq 1 $ \\
		\hline
		$\symbstate_6$ & $\loc_2$ & $1 \leq y - x \leq 2 \land x \geq 0 $ \\
		\hline
		$\symbstate_6'$ & $\loc_2'$ & $1 \leq y - x \leq 2 \land x \geq 0 $ \\
		\hline
		$\symbstate_7$ & $\loc_5$ & $y \geq 2 \land y = x + 1 $ \\
		\hline
		$\symbstate_8$ & $\loc_2$ & $y \geq 2 \land y = x + 2 $ \\
		\hline
		$\symbstate_8'$ & $\loc_2'$ & $y \geq 2 \land y = x + 2 $ \\
		\hline
	\end{tabular}
	
	\caption{Description of the states in \cref{figure:example:MDP}}
	\label{table:zones:description}
\end{table}
\begin{example}
	The probabilistic zone graph of the \PROBTA{} in \cref{figure:example:ProbTA} is the \MDP{} given in \cref{figure:example:MDP}.
	The symbolic states $\symbstate_i = (\loc_i, \C_i)$ are expanded in \cref{table:zones:description}.
\end{example}
\subsection[Reconstructing an \IPROBTA{} from a Probabilistic Zone Graph]{Reconstructing an \IPROBTA{} from a Probabilistic Zone Graph}\label{ss:reconstruction}

It is well-known that, given a timed automata~$\A$ and its zone graph, a second timed automaton~$\A'$ can be reconstructed from the zone graph, with the same structure as the zone graph, and such that the zone graph of~$\A'$ is the same as that of~$\A$.\ea{note pour moi pour un jour : repréciser que ce n'est vrai que lorsque le zone graph vient bien d'un TA, sinon on peut avoir des zones arbitraires qui ne donnent aucun TA…}
We extend this technique here to \IPROBTAs{}.

\paragraph{The construction}
Let $\iprobta{} = (\Actions, \Loc, \locinit, \Clock, \TransitionsIPTA)$ be an \IPROBTA{};
let $\imdp{} = (\SymbStates, \symbstateinit, \TransitionsIPTA, \TransitionsMDP)$ be its probabilistic zone graph.
Let us build a second \IPROBTA{}~$\iprobta' = (\Actions, \Loc', \locinit',\Clock, \TransitionsIPTA')$ as follows.

First, each state of~$\imdp{}$ is translated into a location of~$\iprobta'$, \ie{} we have $\Loc' = \SymbStates$.

Second, the initial location of~$\iprobta'$ is the initial state of~$\imdp$, \ie{} we have $\locinit' = \symbstateinit$.

Third, for each transition $(\symbstate, \pedge, \DistIMDP) \in \TransitionsMDP$ in~$\imdp$, with $\pedge = (\loc,\guard,\action, \DistIPTA)$,
we create in~$\iprobta'$ a transition $(\symbstate,\guard,\action,\DistIPTA')$, where $\DistIPTA'$ is defined as follows:
for each $\symbstate'$ such that $\DistIMDP(\symbstate') > 0$, then $\DistIPTA'(\resets', \symbstate') = \DistIMDP(\symbstate')$, where $\resets'$ is the set of clocks to be reset from $\symbstate.\loc$ to $\symbstate'.\loc$ via edge~$\pedge$ in~$\iprobta$.%

Given an \IPROBTA{} $\iprobta{}$ with probabilistic zone graph $\imdp{}$.
We denote by $\reconstruct(\imdp)$ the \IPROBTA{}~$\iprobta'$ reconstructed from~$\imdp$ following the above technique.
Obviously, this construction also applies to \PROBTA{}, which are just \IPROBTA{} where intervals are reduced to single points.

\paragraph{An equivalence result}
As should be expected, the probabilistic zone graph~$\imdp'$ of the \IPROBTA{}~$\iprobta'$ reconstructed from the probabilistic zone graph~$\imdp$ of a \IPROBTA{}~\iprobta{} is equivalent to~$\imdp$.

\begin{proposition}\label{proposition:reconstruction}
	Let~$\iprobta$ be an \IPROBTA{} and $\imdp$ be its probabilistic zone graph.
	Let $\iprobta' = \reconstruct(\imdp)$.
	Let $\imdp'$ be the probabilistic zone graph of~$\iprobta'$.
	
	Then $\imdp'$ is equivalent to $\imdp$ up to location renaming.\ea{explain?}
\end{proposition}
\begin{figure}
	
	{\centering
		
	\newcommand{\hscale}{2}
	\newcommand{\vscale}{1.5}
	\begin{tikzpicture}[node distance=2cm, auto, ->, >=stealth']
		\node[location0, initial] (l0) at (0,0) {$\symbstate_0$};
		\node [probchoice] (l0choice) at (1*\hscale, 0*\vscale) {};

		\node[location1] (l1) at (2*\hscale, -1*\vscale) {$\loc_1$};
		\node[location2] (l2) at (2*\hscale, +1*\vscale) {$\loc_2$};

		\node[location5] (l5) at (4*\hscale, 0*\vscale) {$\loc_5$};
		\node[location2] (l2b) at (4*\hscale, 1*\vscale) {${\loc_2}_2$};
		\node[location2] (l6) at (4*\hscale, 2*\vscale) {$\loc_2'$};
		
		\node [probchoice] (l2choice) at (3*\hscale, +1*\vscale) {};

		\node[location5] (l5b) at (6*\hscale, 0*\vscale) {${\loc_5}_2$};
		\node[location2] (l6b) at (6*\hscale, 2*\vscale) {${\loc_2'}_2$};
		\node[location2] (l2c) at (6*\hscale, 1*\vscale) {${\loc_2}_3$};
		
		\node [probchoice] (l2bchoice) at (5*\hscale, +1*\vscale) {};
		\node [probchoice] (l6choice) at (5*\hscale, +2*\vscale) {};

                \coordinate[shift={(3mm,6mm)}] (n) at (l6b.north east);
                \coordinate[shift={(10mm,0mm)}] (m) at (l2c.east);

		\path
			(l0) edge node[below]{\begin{tabular}{c}\footnotesize $y < 2$ \\ $a$\end{tabular}} (l0choice) %
		
			(l0choice) edge[probedge] node[below, xshift=-5]{\footnotesize \begin{tabular}{c}$0.6$ \\ $y := 0$\end{tabular}} (l1)
			(l0choice) edge[probedge] node{\footnotesize \begin{tabular}{c}$0.6$ \\ $y := 0$\end{tabular}} (l2)
			
			(l2) edge node[below]{\footnotesize \begin{tabular}{c}$x = 1 \land y \leq 2$ \\ $c$ \end{tabular}} (l2choice) %

			(l2choice) edge[probedge] node[above]{\footnotesize$0.1$} node[below]{\footnotesize$x := 0$} (l2b)
			(l2choice) edge[probedge] node[below, xshift=-3, yshift=1]{\footnotesize$0.8$} (l5)
			(l2choice) edge[probedge] node[above, xshift=3, yshift=1]{\footnotesize \begin{tabular}{c}$0.1$ \\ $x := 0$\end{tabular}} (l6)
			
			(l2b) edge node[below]{\footnotesize \begin{tabular}{c}$x = 1 \land y \leq 2$ \\ $c$ \end{tabular}} (l2bchoice) %

			(l2bchoice) edge[probedge] node[above]{\footnotesize$0.1$} node[below]{\footnotesize$x := 0$} (l2c)
			(l2bchoice) edge[probedge] node[below, xshift=-3]{\footnotesize$0.8$} (l5b)
			(l2bchoice) edge[probedge] node[xshift=10, yshift=-10]{\footnotesize \begin{tabular}{c}$0.1$ \\ $x := 0$\end{tabular}} (l6b)

			(l6) edge node[below]{\footnotesize \begin{tabular}{c}$x = 1 \land y \leq 2$ \\ $c$ \end{tabular}} (l6choice) %

			(l6choice) edge[probedge] node[above, yshift=-4, xshift=4]{\footnotesize \begin{tabular}{c}$0.1$ \\ $x := 0$\end{tabular}} (l6b)

		;

		\draw[->, densely dotted, rounded corners=5pt] (l6choice) to[out = 70] node[above]{\footnotesize$0.9$} (n) to[bend left =20]  (l5b);
		
		\AngleA{l0choice}{l1}{l2};
		
		\AngleA{l2choice}{l6}{l5};
		\AngleA{l2bchoice}{l6b}{l5b};

		\node[draw=none] (l6fake) at (5.5*\hscale, 3.5*\vscale){};
		\AngleA{l6choice}{l6fake}{l6b};
	\end{tikzpicture}
	
	}

	\caption{A \PROBTA{} reconstructed from the probabilistic zone graph in \cref{figure:example:MDP}}
	\label{figure:example:reconstruction}
\end{figure}
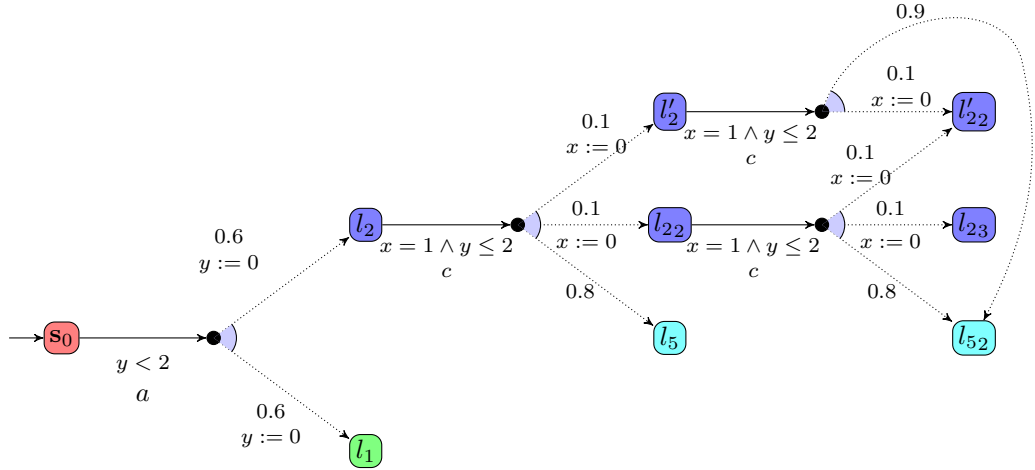
\begin{figure}
	
	{\centering
		
	\newcommand{\hscale}{2}
	\newcommand{\vscale}{1.5}
	\begin{tikzpicture}[node distance=2cm, auto, ->, >=stealth']
		\node[location0, initial] (l0) at (0,0) {$\loc_0$};
		\node [probchoice] (l0choice) at (1*\hscale, 0*\vscale) {};

		\node[location1] (l1) at (2*\hscale, -1*\vscale) {$\loc_1$};
		\node[location2] (l2) at (2*\hscale, +1*\vscale) {$\loc_2$};

		\node [probchoice] (l1choice) at (3*\hscale, -1*\vscale) {};

		\node[location5] (l5) at (4*\hscale, 2*\vscale) {$\loc_5$};
		\node[location2] (l2b) at (4*\hscale, 1*\vscale) {${\loc_2}_2$};
		
		\node [probchoice] (l2choice) at (3*\hscale, +1*\vscale) {};

		\node[location5] (l5b) at (6*\hscale, 2*\vscale) {${\loc_5}_2$};
		\node[location2] (l2c) at (6*\hscale, 1*\vscale) {${\loc_2}_3$};
		
		\node [probchoice] (l2bchoice) at (5*\hscale, +1*\vscale) {};

		\node[location3] (l3) at (5*\hscale, -2*\vscale) {$\loc_3$};
		\node[location4] (l4) at (5*\hscale, 0*\vscale) {$\loc_4$};
		
		\path
			(l0) edge node[above]{\footnotesize $y < 2$} node[below]{\footnotesize $a$} (l0choice) %
		
			(l0choice) edge[probedge] node[below left]{\footnotesize\begin{tabular}{c}$[0,1]$ \\ $y := 0$\end{tabular}} (l1)
			(l0choice) edge[probedge] node{\footnotesize\begin{tabular}{c}$[0,0.5]$ \\ $x := 0$\end{tabular}} (l2)
			
			(l1) edge node[below]{\footnotesize\begin{tabular}{c}$2 \leq x \leq \param$ \\ $b$ \end{tabular}} (l1choice) %

			(l1choice) edge[probedge] node[below left]{\footnotesize\begin{tabular}{c}$[0,0.2]$ \\ $y := 0$\end{tabular}} (l3)
			(l1choice) edge[probedge] node{\footnotesize\begin{tabular}{c}$[0,0.3]$ \\ $x,y := 0$\end{tabular}} (l4)

			(l2) edge node[below]{\footnotesize\begin{tabular}{c}$x = 1 \land y \leq 2$ \\ $c$\end{tabular}} (l2choice) %

			(l2choice) edge[probedge] node[below]{\footnotesize\begin{tabular}{c}$[0,0.2]$ \\ $x := 0$\end{tabular}} (l2b)
			(l2choice) edge[probedge] node{\footnotesize$[0.8,1]$} (l5)
			
			(l2b) edge node[below]{\footnotesize\begin{tabular}{c}$x = 1 \land y \leq 2$ \\ $c$\end{tabular}} (l2bchoice) %

			(l3) edge[bend right] node[right]{\footnotesize\begin{tabular}{c}$x = 5$\\ $d$ \\ $x,y := 0$\end{tabular}} (l4) %
			(l4) edge[bend right] node[left]{\footnotesize\begin{tabular}{c}$2 \leq x \leq \param$\\ $e$ \\ $x := 0$\end{tabular}} (l3) %

			(l2bchoice) edge[probedge] node[below]{\footnotesize\begin{tabular}{c}$[0,0.2]$ \\ $x := 0$\end{tabular}} (l2c)
			(l2bchoice) edge[probedge] node{\footnotesize$[0.8,1]$} (l5b)
			
		;
		
		\AngleA{l0choice}{l1}{l2};
		
		\AngleA{l1choice}{l4}{l3};
		
		\AngleA{l2choice}{l2b}{l5};
		\AngleA{l2bchoice}{l2c}{l5b};

        \end{tikzpicture}
	
	}

	\caption{An \IPROBTA{} reconstructed from the probabilistic zone graph in \cref{figure:example:IMDP}}
	\label{figure:example:reconstructionI}
\end{figure}
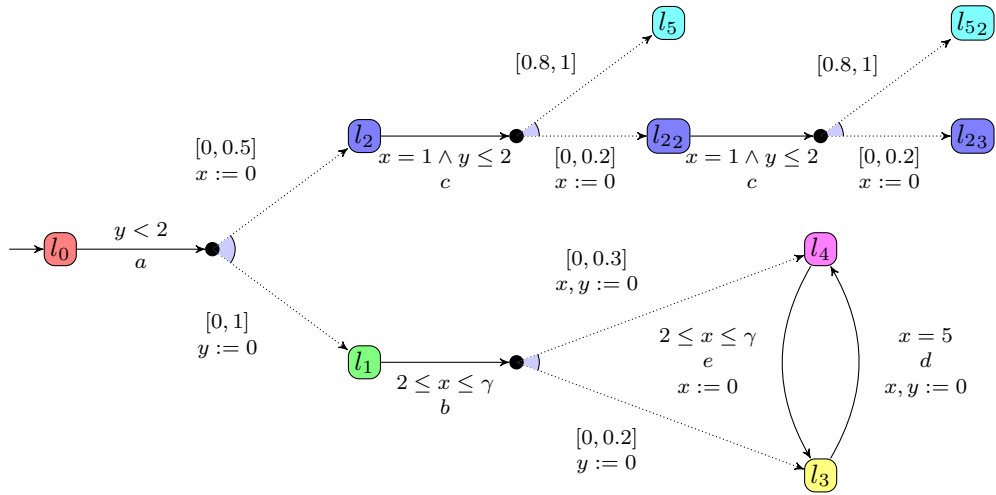

\begin{example}
  We apply the above procedure to the probabilistic zone graphs from
  \cref{figure:example:MDP} and \cref{figure:example:IMDP}. The
  \PROBTA{} and \IPROBTA{} reconstructed from these zone graphs are
  given in \cref{figure:example:reconstruction} and
  \cref{figure:example:reconstructionI}, respectively. Remark that
  their probabilistic zone graphs are again that of
  \cref{figure:example:MDP} and \cref{figure:example:IMDP}.
\end{example}
\subsection{An algorithm for the consistency of \IPROBTAs{}}\label{ss:algo:IPTA}

We start with the following observation: by construction, the purpose
of the symbolic semantics of \IPROBTAs{} is to represent, at a lower
level of abstraction, the same set of objects. Intuitively, the
symbolic \IMDP{} semantics of a given \IPROBTA{} should therefore be consistent
iff the original \IPROBTA{} is itself consistent.
This result is formally proven in \cref{theorem:IPTA:consistency}.

\begin{proposition}\label{theorem:IPTA:consistency}
	An \IPROBTA{}~\iprobta{} is consistent iff its probabilistic zone graph is consistent.
\end{proposition}
\begin{proof}
	Let $\iprobta = (\Actions, \Loc, \locinit, \Clock, \TransitionsIPTA)$ be an \IPROBTA{}. %
	Let $\imdp = (\SymbStates, \symbstateinit, \TransitionsIPTA, \TransitionsMDP)$ be the probabilistic zone graph of~$\iprobta$.
	
	\begin{itemize}
		\item[$\Rightarrow$] Assume \iprobta{} is consistent, and let us show that its probabilistic zone graph is consistent.
			From the definition of consistency, there exists a \PROBTA{}  $\probta = (\Actions, \Loc', \locinit', \Clock, \TransitionsPTA)$ such that $\probta \models \iprobta$, with implementation relation~$\RelSimPTA$.
			Let $\mdp = (\SymbStates', \symbstateinit', \TransitionsPTA, \TransitionsMDP')$ be the probabilistic zone graph of~$\probta$.
			Let us show that $\mdp \models \imdp$.
			
			We therefore define a relation~$\RelSimMDP$, and show that it is an implementation relation.
			We define $\RelSimMDP = \{ ((\loc, \C), (\loc', \C')) \mid (\loc, \loc') \in \RelSimPTA \land \C = \C'\}$.
			
			\begin{itemize}
			 \item From \cref{definition:IPTA:symbolic-semantics}, the initial state of~$\mdp$ is $\symbstateinit' = (\locinit', \timelapse{(\bigwedge_{1 \leq i\leq\ClockCard}\clock_i=0)} )$;
					the initial state of~$\imdp$ is $\symbstateinit = (\locinit, \timelapse{(\bigwedge_{1 \leq i\leq\ClockCard}\clock_i=0)} )$.
					Since $\probta \models \iprobta$ then from \cref{definition:IPROBTA:implementation} we have $(\locinit' , \locinit) \in \RelSimPTA$, and therefore
					$(\symbstateinit' , \symbstateinit) \in \RelSimMDP$.
			\item Let $((\loc, \C), (\loc', \C)) \in \RelSimMDP$.
			\begin{itemize}
				\item Let $((\loc', \C), \action , \DistMDP) \in \TransitionsMDP'$.
					Since \cref{definition:IPTA:symbolic-semantics}, there exists an edge $\pedge' = (\loc', \guard, \action, \DistPTA) \in \TransitionsPTA$.
					Therefore, by~$\RelSimPTA$,there exists an edge  $\pedge = (\loc, \guard, \action, \DistIPTA) \in \TransitionsIPTA$ such that $\DistPTA \preceq_{\RelSimPTA} \DistIPTA$.
					
					As a consequence, by \cref{definition:IPTA:symbolic-semantics} and since the guards are the same in~$\probta$ and~$\iprobta$, there exists $((\loc, \C), \action , \DistIMDP) \in \TransitionsMDP$.
					
					Moreover, by \cref{definition:IPTA:symbolic-semantics}, we have $\DistMDP \preceq_{\RelSimMDP} \DistIMDP$, with \linebreak[4]$\correspondance_{\RelSimMDP}((\loc', \C') , (\loc, \C)) = \correspondance_{\RelSimPTA} (\loc', \loc)$ if $\C = \C'$ and 0 otherwise.
				
				\item Similarly, for all $((\loc, \C), \action , \DistIMDP) \in \TransitionsMDP$, there exists $((\loc', \C), \action , \DistMDP) \in \TransitionsMDP'$ such that $\DistMDP \preceq_{\RelSimMDP} \DistIMDP$ by $\RelSimPTA$ and \cref{definition:IPTA:symbolic-semantics}.
			\end{itemize}
			
			\end{itemize}
		Therefore $\mdp \models \imdp$.

		\item[$\Leftarrow$] Assume the probabilistic zone graph~$\imdp$ of~\iprobta{} is consistent, and let us show that \iprobta{} is consistent.
		From \cref{lemma:structure}, there exists $\mdp = (\SymbStates, \symbstateinit, \TransitionsIPTA, \TransitionsMDP')$ such that $\mdp \models \imdp$ with the same structure as~$\imdp$, and with implementation relation~$\RelSimMDP$ (note that $\RelSimMDP$ is the identity because they have the same structure).
		
		Let us reconstruct an \IPROBTA{}~$\iprobta' = (\Actions, \Loc', \locinit', \Clock, \TransitionsIPTA')$ from the probabilistic zone graph~$\imdp$, using the procedure $\reconstruct$ from \cref{ss:reconstruction}.
		Now, let $\probta = (\Actions, \Loc', \locinit', \Clock, \TransitionsPTA)$ with the same structure as $\iprobta'$, and where $\TransitionsPTA$ is obtained by replacing every occurrence of~$\DistIMDP$ in~$\TransitionsIPTA'$ by $\DistMDP$ taken from~$\TransitionsMDP$ in~$\mdp$.
		Note that there is a one-to-one correspondence between $\TransitionsMDP'$ and $\TransitionsMDP$ since they have the same structure, which is a key point here.

		Recall that, during \reconstruct{}, for each transition $(\symbstate, \pedge, \DistIMDP) \in \TransitionsMDP$ in~$\imdp$, with $\pedge = (\loc,\guard, \action, \DistIPTA)$,
			we create in~$\iprobta'$ a transition $(\symbstate,\guard,\action,\DistIPTA')$, where $\DistIPTA'$ is defined as follows:
			for each $\symbstate'$ such that $\DistIMDP(\symbstate') > 0$, then $\DistIPTA'(\resets', \symbstate') = \DistIMDP(\symbstate')$, where $\resets'$ is the set of clocks to be reset from $\symbstate.\loc$ to $\symbstate'.\loc$ via edge~$\pedge$ in~$\iprobta$.
		Here, we simply replace $\DistIMDP$ with $\DistMDP$, where $\DistMDP$ is the distribution corresponding to~$\DistIMDP$ in~$\mdp$.

		Now, let us show that $\probta \models \iprobta$.
		Recall from \cref{ss:reconstruction} that the locations in~$\probta$ are of the form $(\loc, \C)$.
		We thus define $\RelSimPTA = \{ (\loc', \C') , \loc \mid \loc' = \loc \}$.
		
		\begin{itemize}
			\item From the reconstruction~\reconstruct{}, the initial location of~$\probta$ is $(\locinit, \Cinit)$.
				Therefore, $((\locinit, \Cinit), \locinit) \in \RelSimPTA$.
				
			\item Let $((\loc, \C), \loc') \in \RelSimPTA$.
			\begin{itemize}
				\item Let $\pedge = ((\loc, \C), \guard, \action, \DistPTA) \in \TransitionsPTA$.
					By \reconstruct{}, there must exist $((\loc, \C), \pedge, \DistMDP) \in \TransitionsMDP'$ in~$\mdp$ with $\DistPTA(\resets, (\loc', \C')) = \DistMDP((\loc', \C'))$, for all $(\loc', \C')$, where $\resets$ is the set of clocks to be reset from $\loc$ to $\loc'$ via edge~$\pedge$ in~$\probta$.
					Therefore, from $\RelSimMDP$, there exists $((\loc, \C), \pedge, \DistIMDP) \in \TransitionsMDP$ of~$\imdp$ \st{} $\DistMDP \preceq_{\RelSimMDP} \DistIMDP$.
					
					As a consequence, by the zone graph construction (\cref{definition:IPTA:symbolic-semantics}), there exists a transition $(\loc, \guard, \action, \DistIPTA) \in \TransitionsIPTA$ in $\iprobta$ such that $\DistIMDP((\loc', \C')) = \DistIPTA(\resets, \loc')$, for all $\resets, \loc', \C'$.
					
					Let $\delta_{\RelSimPTA}$ be such that $\delta_{\RelSimPTA}((\loc', \C') , \loc'') = 1$ if $\loc' = \loc''$ and 0 otherwise.
					Finally, by $\RelSimMDP$, we obtain $\DistPTA \preceq_{\RelSimPTA} \DistIPTA$.
					\ea{revoir :-)}
					
				\item Similarly, %
					for all $(\loc, \guard, \action, \DistIPTA) \in \TransitionsIPTA$ in $\iprobta$, there exists \linebreak[4]$((\loc, \C), \guard, \action, \DistPTA) \in \TransitionsPTA$ such that $\DistPTA \preceq_{\RelSimPTA} \DistIPTA$.

			\end{itemize}

		\end{itemize}
		Therefore $\probta \models \iprobta$.

	\end{itemize}
\end{proof}

Given the results presented in \cref{lemma:structure} and
\cref{theorem:IPTA:consistency}, deciding whether a
given \IPROBTA{}~\iprobta{} is consistent can be done by deciding
whether its probabilistic zone graph admits at least one implementation that
preserves its structure.

Such an algorithm was provided in~\cite{BDSyncop15} in the context of
IMCs instead of \IMDPs{}. We show how this
algorithm can be adapted to our context. As for IMCs, we say that a
state is {\em locally inconsistent} in a given \IMDP{} iff one of its
outgoing probabilistic (interval) transitions cannot be implemented,
\ie{} if there is no distribution that matches the specified
intervals. Let $\imdp = (\States, \stateinit, \TransitionsIPTA, \TransitionsMDP)$
be the \IMDP{} symbolic semantics of a given \IPROBTA{}.
Our algorithm is given in \cref{algo:iprobta}.

\begin{algorithm}
Let $\IncStates$ be the set of locally inconsistent states in $\imdp$ and $\PassedStates = \emptyset$.

\While{$\state_0 \notin \PassedStates$ and $\IncStates \ne \emptyset$}{
	Let $\state \in \IncStates$ and $\PassedStates = \PassedStates \cup \{\state\}$.
	
	Replace all transitions $(\state',a,\DistIMDP)$ such that
        $\DistIMDP(\state) \ne [0,0]$ with $(\state',a,\DistIMDP')$ where
        \begin{itemize} 
        \item $\DistIMDP'(\state'') = \DistIMDP(\state'')$ for all $\state'' \ne \state$, and
         \item $\DistIMDP'(\state)=\DistIMDP(\state)\cap[0,0]$
        \end{itemize} 
        
	Update $\IncStates \subseteq (\States \setminus \PassedStates)$.
} %
\caption{Consistency of \IMDPs{}}
\label{algo:iprobta}
\end{algorithm}

\cref{algo:iprobta} is based on the following principle: as soon as a locally
inconsistent state is detected, it is made unreachable by
modifying the incoming interval probabilities to $\DistIMDP(\state)\cap[0,0]$.
Remark that if $0$ is not an
admissible transition probability, the inconsistency is transfered to the predecessor
states because $\DistIMDP(\state)\cap[0,0]=\emptyset$.

In the context of IMCs, it is proven in~\cite{BDSyncop15} that this
algorithm converges and that the original IMC is consistent iff the
initial state is not locally inconsistent in the resulting IMC. The
proof from~\cite{BDSyncop15} can be trivially adapted to the context
of \IMDPs{}.

\cref{theorem:IPTA:consistency} together with \cref{algo:iprobta} and the above discussion on termination give the following theorem:

\begin{theorem}\label{theorem:IPROBTA-decidable}
	The consistency problem for \IPROBTAs{} is decidable.
\end{theorem}
\section{Consistency-emptiness and synthesis for \PIPROBTAs{}}\label{section:consistency-PIPROBTA}

We now move to the parametric setting and consider the following two problems:

\defProblem
	{Consistency-emptiness}
	{A \PIPROBTA{}~$\piprobta$}
	{does there exist a parameter valuation $\pval$ such that $\valuate{\piprobta}{\pval}$ is consistent?}

\defProblem
	{Consistency-synthesis}
	{A \PIPROBTA{}~$\piprobta$}
	{find all parameter valuations $\pval$ for which $\valuate{\piprobta}{\pval}$ is consistent.}

In the following, we first address the consistency-emptiness problem
and show that, while this problem is undecidable in the general context of \PIPROBTAs{}, it becomes decidable for a syntactic subclass (\cref{ss:emptiness}).
We then introduce an adaptation of the parametric zone-graph construction for parametric timed automata (\cref{ss:symbolic-semantics}), and propose in \cref{ss:synthesis} a construction for the consistency-synthesis problem.
This construction can only be applied when the parametric probabilistic zone-graph construction of the original \PIPROBTA{} is finite.
When this is the case, the set of parameter values that are synthesized is exactly those that ensure consistency of the resulting \IPROBTA{}.
We finally address the more general problem of consistent reachability in \cref{ss:consistent-reachability}.

\subsection[The emptiness problem]{The emptiness problem}\label{ss:emptiness}
\subsubsection[Undecidability in the general case]{Undecidability in the general case}\label{sss:undecidable}

The undecidability of the consistency-emptiness for \PIPROBTAs{} follows from the undecidability of the reachability emptiness for parametric timed automata.

\begin{theorem}\label{theorem:undecidability}
	The consistency-emptiness for \PIPROBTAs{} is undecidable.
\end{theorem}
\begin{proof}
	The reachability emptiness for parametric timed automata (\ie{} the existence of at least one parameter valuation for which a given location is reachable) is undecidable (see \eg{} \cite{AHV93,JLR15,AM15,BBLS15}, and \cite{Andre18STTT} for a complete survey).
	In particular, it is undecidable even without invariant as shown in~\cite{AHV93,BBLS15}, which is inline with our setting.
	
	We prove our result by reducing from the reachability emptiness for parametric timed automata.
	Assume a PTA (without probability), with a special goal location.
	From that goal location, let us add an unguarded transition to a new location for which no implementation exists (for example a single transition labeled with $[0.5,0.5]$).
	Hence there exists a parameter valuation for which the underlying \IPROBTA{} admits no implementation iff there exists a parameter valuation for which the goal location is reachable---which is undecidable.
\end{proof}

The undecidability of the consistency-emptiness problem rules out the possibility to, in general, compute a solution to the consistency-synthesis problem.
In the following, we will still address this computation problem by
proposing a synthesis procedure that can be used when the parametric
probabilistic zone graph is finite.

\subsubsection[A decidability result]{A decidability result}\label{sss:decidable}

Despite the negative result of \cref{theorem:undecidability}, we can exhibit a decidability result for a syntactic subclass of \PIPROBTAs{}.
In~\cite{HRSV02}, a syntactic subclass, namely lower-bound/upper-bound parametric timed automata (\emph{L/U-PTAs}) is introduced that restricts the use of parameters in parametric timed automata.
Basically, in an L/U-PTA, any parameter must be either always used as an upper-bound (a constraint $\clock \leq \param$ or $\clock < \param$) or always as a lower-bound ($\clock \geq \param$ or $\clock > \param$) in the guards and invariants of the parametric timed automaton.
L/U-PTAs benefit from several main decidability results:
	the \emph{\EF{}-emptiness}, or \emph{reachability-emptiness}, problem (``is the set of parameter valuations for which a given location is reachable empty?'')\ is shown to be decidable in~\cite{HRSV02}.
Then, the infinite acceptance emptiness (``is the set of parameter valuations for which a given set of locations is visited infinitely often along some run empty?'')\ and universality (``is a given set of locations visited infinitely often along some run for all parameter valuations?'')\ have been proved to be decidable for L/U-PTAs with integer-valued parameters in~\cite{BlT09}.
Unavoidability was studied in~\cite{JLR15} while liveness and deadlocks were studied in~\cite{ALime17} with a thin frontier between decidability and undecidability.
Finally, the \EF{}-universality problem was shown to be decidable for L/U-PTAs over rational parameters by~\cite{Andre18STTT}.

In the following, we reuse the concept of lower-bound and upper-bound parameters in the setting of \PIPROBTAs{}.

\begin{definition}[L/U-\PIPROBTA{}]\label{def:LUPIPTA}
	An L/U-\PIPROBTA{} is a \PIPROBTA{} whose set of parameters is partitioned into lower-bound parameters and upper-bound parameters,
	where an upper-bound (resp.\ lower-bound) parameter~$\param_i$ is such that, 
    for every guard %
    constraint $\clock \compOp z$, we have: $z = \param_i$ implies ${\compOp} \in \{ \leq, < \}$ (resp.\ ${\compOp} \in \{ \geq, > \}$).
\end{definition}
\begin{example}\label{example:LU}
	Consider the \PIPROBTA{} in \cref{figure:LU:PIPROBTA}.
	Then it is an L/U-\PIPROBTA{} with one upper-bound parameter $\param_1$ and one lower-bound parameter~$\param_2$.
\end{example}

L/U-PTAs enjoy a well-known monotonicity property recalled in the following lemma (that corresponds to a reformulation of \cite[Prop~4.2]{HRSV02}), stating that increasing upper-bound parameters or decreasing lower-bound parameters can only add behaviors.

\begin{lemma}[\cite{HRSV02}]\label{lemma:HRSV02:prop4.2}
	Let~$\A$ be an L/U-PTA and~$\pval$ be a parameter valuation.
	Let $\pval'$ be a valuation such that
	for each upper-bound parameter~$\param^+$, $\pval'(\param^+) \geq \pval(\param^+)$
	and
	for each lower-bound parameter~$\param^-$, $\pval'(\param^-) \leq \pval(\param^-)$.
	Then any run of~$\valuate{\A}{\pval}$ is a run of $\valuate{\A}{\pval'}$. 
\end{lemma}
\begin{remark}\label{remark:cannot:LU}
	Unfortunately, we cannot directly use the monotonicity property of L/U-PTAs to prove the decidability of the consistency-emptiness for L/U-\PIPROBTAs{}.
	It could have been handful to consider the \IPROBTA{}, say $\iprobtainfzero$, obtained by replacing every lower-bound parameter (resp.\ upper-bound parameter) in the guards of a \PIPROBTA{}~$\piprobta$ with~$\infty$ (resp.\ 0).\footnote{%
		Valuating a parameter with~$\infty$ is achieved as follows:
		 for each upper-bound parameter~$\param$ for which $\pval(\param) = \infty$, we delete any comparison of a clock with~$\param$ (\ie{} the clock constraint becomes the most permissive);
		for each lower-bound parameter~$\param$ for which $\pval(\param) = \infty$, we replace any constraint in which $\param$ appears by $\BFalse$ (\ie{} the transition labeled by the guard is deleted).
		Therefore the result of this valuation is an \IPROBTA{} as expected.
	}
		$\iprobtainfzero$ is the most restrictive \IPROBTA{} obtained from~$\piprobta$ (as any other valuation will have more behaviors thanks to~\cref{lemma:HRSV02:prop4.2}).
	If $\iprobtainfzero$ is inconsistent, then any other parameter valuation is clearly inconsistent as well thanks to the monotonicity of L/U-PTAs; however, if $\iprobtainfzero$ is consistent, it is not possible to conclude that a (non-infinite) parameter valuation is consistent.
	In fact, it is easy to exhibit a counter-example for which $\iprobtainfzero$ is consistent, but for which $\valuate{\piprobta}{\pval}$ is inconsistent for any (non-infinite) valuation~$\pval$.
	This is the case of the \PIPROBTA{} depicted in \cref{figure:counter-example-LU}: when $\param$ is replaced with $\infty$, the system is stuck forever in~$\loc_0$, and is therefore consistent.
	For any (non-infinite) parameter valuation, the system can take the transition labeled with~$a$, which is inconsistent due to the sum of the probabilities.
\end{remark}
\begin{figure}
	
	{\centering
		
	\newcommand{\hscale}{2}
	\newcommand{\vscale}{1.5}
	\begin{tikzpicture}[node distance=2cm, auto, ->, >=stealth']
		\node[location0, initial] (l0) at (0,0) {$\loc_0$};
		\node [probchoice] (l0choice) at (1*\hscale, 0*\vscale) {};

		\node[location1] (l1) at (2*\hscale, -.5*\vscale) {$\loc_1$};
		\node[location2] (l2) at (2*\hscale, +.5*\vscale) {$\loc_2$};

		\path
			(l0) edge node[below]{\begin{tabular}{c}\footnotesize $\clock \geq \param$ \\ $a$\end{tabular}} (l0choice) %
		
			(l0choice) edge[probedge] node[below, xshift=-5]{\footnotesize \begin{tabular}{c}$0.6$ \\ $\clock := 0$\end{tabular}} (l1)
			(l0choice) edge[probedge] node[above]{\footnotesize \begin{tabular}{c}$0.9$ \\ $\clock := 0$\end{tabular}} (l2)
			
		;

		\AngleA{l0choice}{l1}{l2};

		\end{tikzpicture}
	
	}

	\caption{An L/U-\PIPROBTA{} inconsistent for all parameter valuations}
	\label{figure:counter-example-LU}
\end{figure}
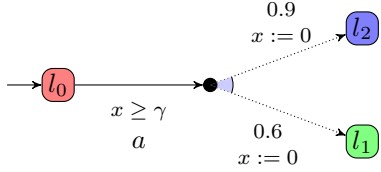

Still, we show in \cref{theorem:decidability-LU} below that the consistency-emptiness problem is decidable in the context of L/U-\PIPROBTAs{}.

\todo{il y a un raisonnement mieux au début de la 4.3}
To prove decidability, we use the following reasoning.
An L/U-\PIPROBTA{} is consistent if we can ``block'' the inconsistent edges, \ie{} those who cannot admit any implementation because the sum of their probabilities cannot be~1.
There are two ways of achieving this goal: either set to~0 some of the probabilities on all paths leading to a given inconsistent edge, or tune the timing parameters so as to forbid this edge because the guard can never be satisfied.
The first way can be achieved by enumerating all possible combinations to set to~0 some probabilities.
The second way can be achieved by parametric model checking: for a given combination of probabilities set to~0, if we can find at least one valuation for which none of the inconsistent edges is reachable, then we can answer false to the consistency-emptiness problem for L/U-\PIPROBTAs{}.
Finding at least one such valuation is equivalent to answering no to the \EF-universality problem---which is decidable for L/U-PTAs as shown in~\cite{Andre18STTT}.
In the following, we explain this reasoning step by step.

We first need a notation, used in the proof of \cref{theorem:decidability-LU}, and later in \cref{ss:synthesis}.

\begin{definition}[feasible supports]\label{definition:FS}
	Given an interval distribution~$\DistIMDP$ in an \IMDP{}, let $\CombiInt(\DistIMDP)$ denote the \emph{feasible supports} of $\DistIMDP$ \ie{}
	set of sets of target states for which a consistent distribution can be assigned.
	Formally, $\CombiInt(\DistIMDP) = \{ \States' \subseteq \States \mid 
	\exists \mu \in \dist(\States) \text{ s.t. }
	\forall \state\in \States : \mu(\state)\in \DistIMDP(\state)\text{ and } 
	\forall \state\in \States : \mu(\state)>0 \text{ iff } \state\in\States' \}$.
\end{definition}

That is, each set of states~$\States'$ is such that there exists a distribution~$\DistIMDP'$ for which the probability of reaching each state in~$\States'$ is not zero, such that this distribution is a punctual distribution, is an implementation of~$\DistIMDP$ and is consistent.

\begin{definition}\label{definition:inconsistent-edge}
	An interval distribution~$\DistIMDP$ is \emph{inconsistent} if $\CombiInt(\DistIMDP) = \emptyset$.
\end{definition}

By extension, we say that an edge is inconsistent if its interval distribution is inconsistent.

We also use the same notions for interval distributions in~\PIPROBTAs{}.

\begin{example}\label{example:FS}
	In \cref{figure:example:PIPROBTA}, let~$\DistIPTA_1$ be the (unique) interval distribution outgoing from~$\loc_1$.
	We have $\CombiInt(\DistIPTA_1) = \{ \}$, as no implementation can make this distribution consistent.
	That is, $\DistIPTA_1$ is an inconsistent edge.
	Let~$\DistIPTA_2$ be the (unique) interval distribution outgoing from~$\loc_2$.
	We have $\CombiInt(\DistIPTA_2) = \{ \{ \loc_5 \}, \{ \loc_2 , \loc_5 \} \}$.
\end{example}

We now define the set of \PIPROBTAs{} obtained by taking all possible combinations of feasible supports.

\begin{definition}\label{definition:PIPFS}
	Given a \PIPROBTA{}~$\piprobta$, let $\PIPCombiInt(\piprobta)$ denote the set of all possible \PIPROBTAs{} obtained from~$\piprobta$ by selecting for each interval distribution~$\DistIPTA$ exactly one element from~$\CombiInt(\DistIPTA)$ when $\CombiInt(\DistIPTA) \neq \emptyset$, or by keeping the original distribution if $\CombiInt(\DistIPTA) = \emptyset$.
\end{definition}

Intuitively, $\PIPCombiInt(\piprobta)$ contains all possible ways to remove transitions by setting some probabilities to~0 while keeping the sum of the other probabilities possibly equal to~1.

\begin{example}\label{example:PIPFS}
	Consider the \PIPROBTA{}~$\piprobta$ in \cref{figure:example:PIPROBTA}.
	$\PIPCombiInt(\piprobta)$ contains 4~\PIPROBTAs{}.
	All are such that~$\loc_1$ has the same outgoing distribution to $\loc_4$ and~$\loc_3$ as in \cref{figure:example:PIPROBTA} (as $\CombiInt$ is empty for this distribution).
	Two of these 4~\PIPROBTAs{} (say 1 and~2) are such that the distribution outgoing from $\loc_0$ goes to both $\loc_1$ and $\loc_2$ (with the same probabilities as in \cref{figure:example:PIPROBTA}), while two others (say 3 and~4) are such that this distribution is only going to~$\loc_1$.
	In addition, two of these 4~\PIPROBTAs{} (say 1 and~3) are such that the distribution outgoing from $\loc_2$ goes to both $\loc_2$ and $\loc_5$, while two others (say 2 and~4) are such that this distribution is only going to~$\loc_5$.
\end{example}
\begin{example}\label{example:PIPFS2}
	Consider now the \PIPROBTA{}~$\piprobta$ in \cref{figure:LU:PIPROBTA}.
	Then $\PIPCombiInt(\piprobta)$ contains the 2~\PIPROBTAs{} in \cref{figure:LU:PIPROBTA} (\ie{} itself) and in \cref{figure:LU:PIPCombiInt}.
\end{example}

Given an L/U-PTA~$\A$ and a subset~$\somelocs$ of its locations, let us denote by $\EFuniv(\A, \somelocs)$ the result of \EF-universality for locations~$\somelocs$ in~$\A$, \ie{} the answer to the following question: ``is the set of valuations~$\pval$ such that at least one location of~$\somelocs$ is reachable in $\valuate{\A}{\pval}$ universal?''
Or put differently, do \emph{all} valuations reach at least one location of~$\somelocs$?
Recall that \EF-universality is decidable for L/U-PTAs as shown in~\cite{Andre18STTT}.

We need two additional notations before introducing our decision procedure.
First, given a \PIPROBTA{}~$\piprobta$, let $\makeNonDet(\piprobta)$ denote the PTA obtained from~$\piprobta$ by performing the following three operations:
\begin{enumerate}
	\item\label{item:addedges} for each location~$\loc$, for each inconsistent edge~$\DistIPTA$ from~$\loc$, add a new non-probabilistic transition from~$\loc$ to a fresh location, with the same guard as on $\DistIPTA$;
	\item remove all inconsistent edges;
	\item replace all probabilistic distributions with non-determinism (see \eg{} \cite{AFS13}).
\end{enumerate}
Observe that, if~$\piprobta$ is an L/U-\PIPROBTA{}, then $\makeNonDet(\piprobta)$ is an L/U-PTA.
Second, $\makeAcc(\piprobta)$ returns the set of locations %
added at step~\ref{item:addedges}.

Beyond transforming the L/U-\PIPROBTA{} into a non-probabilistic L/U-PTA, the rationale behind $\makeNonDet$ is that we need to test for reachability of an \emph{edge}, which is not possible natively in L/U-PTAs; therefore, we add new locations target of these edges.
In addition, we cannot test for the reachability of an arbitrary existing location target of these edges, as they may be reached via other paths too.
In \cref{figure:LU:PIPROBTA}, $\loc_3$ is reachable via the inconsistent edge outgoing from~$\loc_1$, but also directly from~$\loc_0$: only reaching~$\loc_3$ via~$\loc_1$ should be avoided, which justifies the creation of~$\loc_1'$ in \cref{figure:LU:makeNonDet}.

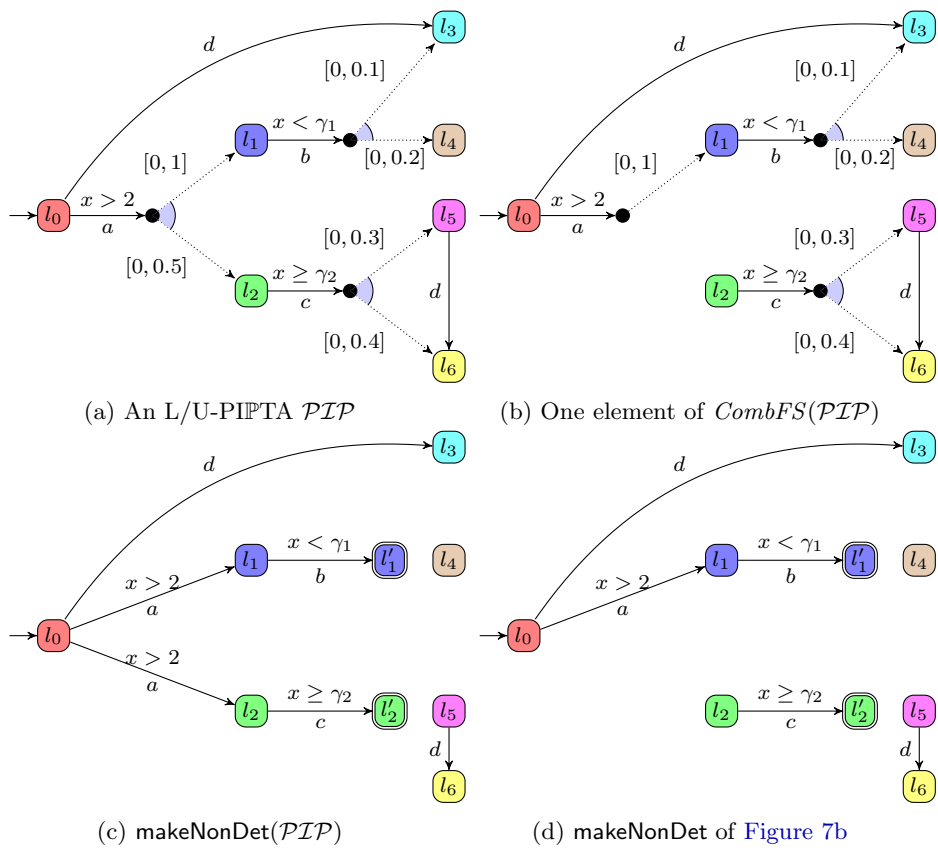
\begin{figure}
	\footnotesize

	\newcommand{\hscale}{1.3}
	\newcommand{\vscale}{1}

	\begin{subfigure}[b]{.49\textwidth}
	{%
	
	\begin{tikzpicture}[node distance=2cm, auto, ->, >=stealth']
		\node[location0, initial] (l0) at (0,0) {$\loc_0$};
		\node [probchoice] (l0choice) at (1*\hscale, 0*\vscale) {};

		\node[location2] (l1) at (2*\hscale, +1*\vscale) {$\loc_1$};
		\node[location1] (l2) at (2*\hscale, -1*\vscale) {$\loc_2$};

		\node [probchoice] (l1choice) at (3*\hscale, +1*\vscale) {};
		\node [probchoice] (l2choice) at (3*\hscale, -1*\vscale) {};

		\node[location5] (l3) at (4*\hscale, 2.5*\vscale) {$\loc_3$};
		\node[location6] (l4) at (4*\hscale, 1*\vscale) {$\loc_4$};

		\node[location4] (l5) at (4*\hscale, 0*\vscale) {$\loc_5$};
		\node[location3] (l6) at (4*\hscale, -2*\vscale) {$\loc_6$};
		
		\path
			(l0) edge node[above]{$\clock > 2$} node[below]{$a$} (l0choice)
		
			(l0choice) edge[probedge] node{$[0, 1]$} (l1)
			(l0choice) edge[probedge] node[below left]{$[0, 0.5]$} (l2)
			
			(l1) edge node[above]{$\clock < \param_1$} node[below]{$b$} (l1choice)

			(l2) edge node[above]{$\clock \geq \param_2$} node[below]{$c$} (l2choice)

			(l1choice) edge[probedge] node{$[0,0.1]$} (l3)
			(l1choice) edge[probedge] node[below]{$[0,0.2]$} (l4)
			
			(l2choice) edge[probedge] node{$[0,0.3]$} (l5)
			(l2choice) edge[probedge] node[below left]{$[0,0.4]$} (l6)

			(l0) edge[bend left] node[above left]{$d$} (l3)
			(l5) edge node[left]{$d$} (l6)
		;
		
		\AngleA{l0choice}{l2}{l1};
		
		\AngleA{l2choice}{l5}{l6};
		
		\AngleA{l1choice}{l4}{l3};

	\end{tikzpicture}
	
	\caption{An L/U-\PIPROBTA{}~$\piprobta$}
	\label{figure:LU:PIPROBTA}
	}
	\end{subfigure}
	\hfill{}
	\begin{subfigure}[b]{.49\textwidth}
	{%
	
	\begin{tikzpicture}[node distance=2cm, auto, ->, >=stealth']
		\node[location0, initial] (l0) at (0,0) {$\loc_0$};
		\node [probchoice] (l0choice) at (1*\hscale, 0*\vscale) {};

		\node[location2] (l1) at (2*\hscale, +1*\vscale) {$\loc_1$};
		\node[location1] (l2) at (2*\hscale, -1*\vscale) {$\loc_2$};

		\node [probchoice] (l1choice) at (3*\hscale, +1*\vscale) {};
		\node [probchoice] (l2choice) at (3*\hscale, -1*\vscale) {};

		\node[location5] (l3) at (4*\hscale, 2.5*\vscale) {$\loc_3$};
		\node[location6] (l4) at (4*\hscale, 1*\vscale) {$\loc_4$};

		\node[location4] (l5) at (4*\hscale, 0*\vscale) {$\loc_5$};
		\node[location3] (l6) at (4*\hscale, -2*\vscale) {$\loc_6$};
		
		\path
			(l0) edge node[above]{$\clock > 2$} node[below]{$a$} (l0choice)
		
			(l0choice) edge[probedge] node{$[0, 1]$} (l1)

			(l1) edge node[above]{$\clock < \param_1$} node[below]{$b$} (l1choice)

			(l2) edge node[above]{$\clock \geq \param_2$} node[below]{$c$} (l2choice)

			(l1choice) edge[probedge] node{$[0,0.1]$} (l3)
			(l1choice) edge[probedge] node[below]{$[0,0.2]$} (l4)
			
			(l2choice) edge[probedge] node{$[0,0.3]$} (l5)
			(l2choice) edge[probedge] node[below left]{$[0,0.4]$} (l6)

			(l0) edge[bend left] node[above left]{$d$} (l3)
			(l5) edge node[left]{$d$} (l6)
		;

		\AngleA{l2choice}{l5}{l6};
		
		\AngleA{l1choice}{l4}{l3};

	\end{tikzpicture}
	
	\caption{One element of $\PIPCombiInt(\piprobta)$}
	\label{figure:LU:PIPCombiInt}
	}
	\end{subfigure}
	\begin{subfigure}[b]{.49\textwidth}
	{%

	\begin{tikzpicture}[node distance=2cm, auto, ->, >=stealth']
		\node[location0, initial] (l0) at (0,0) {$\loc_0$};

		\node[location2] (l1) at (2*\hscale, +1*\vscale) {$\loc_1$};
		\node[location1] (l2) at (2*\hscale, -1*\vscale) {$\loc_2$};

		\node [location2,final] (l1choice) at (3.4*\hscale, +1*\vscale) {$\loc_1'$};
		\node [location1,final] (l2choice) at (3.4*\hscale, -1*\vscale) {$\loc_2'$};

		\node[location5] (l3) at (4*\hscale, 2.5*\vscale) {$\loc_3$};
		\node[location6] (l4) at (4*\hscale, 1*\vscale) {$\loc_4$};
		
		\node[location4] (l5) at (4*\hscale, -1*\vscale) {$\loc_5$};
		\node[location3] (l6) at (4*\hscale, -2*\vscale) {$\loc_6$};
		
		\path

			(l0) edge node[above]{$\clock > 2$} node[below]{$a$} (l1)
			(l0) edge node[above]{$\clock > 2$} node[below]{$a$} (l2)

			(l1) edge node[above]{$\clock < \param_1$} node[below]{$b$} (l1choice)
			(l2) edge node[above]{$\clock \geq \param_2$} node[below]{$c$} (l2choice)

			(l0) edge[bend left] node[above left]{$d$} (l3)
			(l5) edge node[left]{$d$} (l6)
		;

	\end{tikzpicture}
	
	\caption{$\makeNonDet(\piprobta)$}
	\label{figure:LU:makeNonDet}
	}
	\end{subfigure}
	\hfill{}
	\begin{subfigure}[b]{.49\textwidth}
	{%

	\begin{tikzpicture}[node distance=2cm, auto, ->, >=stealth']
		\node[location0, initial] (l0) at (0,0) {$\loc_0$};

		\node[location2] (l1) at (2*\hscale, +1*\vscale) {$\loc_1$};
		\node[location1] (l2) at (2*\hscale, -1*\vscale) {$\loc_2$};

		\node [location2,final] (l1choice) at (3.4*\hscale, +1*\vscale) {$\loc_1'$};
		\node [location1,final] (l2choice) at (3.4*\hscale, -1*\vscale) {$\loc_2'$};

		\node[location5] (l3) at (4*\hscale, 2.5*\vscale) {$\loc_3$};
		\node[location6] (l4) at (4*\hscale, 1*\vscale) {$\loc_4$};
		
		\node[location4] (l5) at (4*\hscale, -1*\vscale) {$\loc_5$};
		\node[location3] (l6) at (4*\hscale, -2*\vscale) {$\loc_6$};
		
		\path

			(l0) edge node[above]{$\clock > 2$} node[below]{$a$} (l1)

			(l1) edge node[above]{$\clock < \param_1$} node[below]{$b$} (l1choice)

			(l2) edge node[above]{$\clock \geq \param_2$} node[below]{$c$} (l2choice)

			(l0) edge[bend left] node[above left]{$d$} (l3)
			(l5) edge node[left]{$d$} (l6)
		;

	\end{tikzpicture}
	
	\caption{$\makeNonDet$ of \cref{figure:LU:PIPCombiInt}}
	\label{figure:LU:makeNonDet+PIPCombiInt}
	}
	\end{subfigure}

	\caption{Exemplifying $\PIPCombiInt$, $\makeNonDet$ and $\makeAcc$}
	\label{figure:examples:LU}
\end{figure}
\begin{example}\label{example:makeNonDet}
	Consider again the \PIPROBTA{}~$\piprobta$ in \cref{figure:LU:PIPROBTA}.
	Then the L/U-PTA result of $\makeNonDet(\piprobta)$ is given in \cref{figure:LU:makeNonDet}, while $\makeAcc(\piprobta)$ is $\{ \loc_1', \loc_2' \}$.
\end{example}

We can now give below the main equation to solve consistency-emptiness for L/U-\PIPROBTAs{}.

\begin{equation}\label{equation:LU}
	\bigwedge_{\piprobta' \in \PIPCombiInt(\piprobta)} \EFuniv\big(\makeNonDet(\piprobta'), \makeAcc(\piprobta') \big)
\end{equation}

The idea is that consistency-emptiness holds for an L/U-\PIPROBTA{} if, for each combination of probabilities set to~0 ($\PIPCombiInt$), for all parameter valuations (\EFuniv{}), some of the locations target of an inconsistent distribution ($\makeAcc$) are always reachable.
In other words, there is no way to set some probabilities to~0 and to exhibit some parameter valuations that would avoid an inconsistent distribution.

Note that this procedure can be easily implemented by enumerating all \PIPROBTAs{} in $\PIPCombiInt(\piprobta)$, replacing probabilities with non-determinism as in~$\makeNonDet$, and testing \EFuniv{} on each resulting L/U-PTA using the procedures given in \cite{BlT09,Andre18STTT}.

\begin{example}\label{example:LU:procedure}
	Consider again the \PIPROBTA{}~$\piprobta$ in \cref{figure:LU:PIPROBTA}.
	Recall that $\PIPCombiInt(\piprobta)$ is given in \cref{figure:LU:PIPROBTA,figure:LU:PIPCombiInt},
	and $\makeAcc(\piprobta) = \{ \loc_1', \loc_2' \}$.
	Therefore, checking consistency-emptiness for~$\piprobta$ amounts to checking \EF{}-universality of locations~$\{ \loc_1', \loc_2' \}$ in the L/U-PTAs in \cref{figure:LU:makeNonDet,figure:LU:makeNonDet+PIPCombiInt}.
	For the L/U-PTA in \cref{figure:LU:makeNonDet}, $\EFuniv$ gives true, as $\loc_2'$ can be reached for any valuation of~$\param_2$ and regardless of~$\param_1$ (it suffices to wait enough time in~$\loc_2$ so that the guard $\clock \geq \param_2$ becomes enabled).
	However, for the L/U-PTA in \cref{figure:LU:makeNonDet+PIPCombiInt}, $\EFuniv$ gives false: indeed, while $\loc_2'$ is clearly unreachable, $\loc_1'$ can only be reached if $\param_1 > 2$.
	Therefore, there exist valuations (typically $\param_1 \in [0,2]$) for which locations~$\{ \loc_1', \loc_2' \}$ are unreachable.
	
	In fact, it can be shown that, for the \PIPROBTA{}~$\piprobta$ in \cref{figure:LU:PIPROBTA}, the set of valuations for which the \IPROBTA{} is consistent is $\param_1 \in [0,2] \land \param_2 \geq 0$.
	The idea is to disable the transition to~$\loc_2$ using probabilities (\ie{} assigning 1 to~$\loc_1$ and 0 to~$\loc_2$), and to disable the transitions to~$\loc_3$ and~$\loc_4$ by tuning~$\param_1$.
\end{example}
\begin{theorem}\label{theorem:decidability-LU}
	The consistency-emptiness for L/U-\PIPROBTAs{} is decidable.
\end{theorem}
\begin{proof}
	Given an L/U-\PIPROBTA{}~$\piprobta$, we show that Equation~\ref{equation:LU} holds iff the consistency-emptiness holds for~$\piprobta$, \ie{} no parameter valuation~$\pval$ is such that $\valuate{\piprobta}{\pval}$ is consistent.
	\begin{itemize}
		\item[$\Rightarrow$]
			Assume Equation~\ref{equation:LU} holds.
			Then EF-universality is true for all possible combinations of probabilities set to~0 (given by $\PIPCombiInt(\piprobta)$).
			That is, for each of these potentially consistent models, for any valuation~$\pval$, it is always possible to reach at least one of the new locations added by $\makeNonDet$, and therefore one of the inconsistent edges in the original model.
			Therefore, from \cref{definition:consistency:IPROBTA}, $\valuate{\piprobta}{\pval}$ is inconsistent for all~$\pval$.
			Therefore the consistency-emptiness holds for~$\piprobta$.
		\item[$\Leftarrow$]
			Assume consistency-emptiness holds for~$\piprobta$, \ie{} no parameter valuation~$\pval$ is such that $\valuate{\piprobta}{\pval}$ is consistent.
			Then there is no way to tune the probabilities and to tune the timing parameters to avoid the inconsistent edges, and therefore to avoid the new locations added by $\makeNonDet$.
			Then for any $\piprobta' \in \PIPCombiInt(\piprobta)$, we have that
			\[\EFuniv(\makeNonDet(\piprobta'), \makeAcc(\piprobta'))\]
			holds.
			Then Equation~\ref{equation:LU} holds.
	\end{itemize}
	The result then follows from the decidability of the EF-universality problem for L/U-PTAs proved in~\cite{Andre18STTT}.
\end{proof}

\todo{un jour: Observe that, following a dual result, the decidability of the consistency-universality problem for L/U-\PIPROBTAs{} can also be shown.\ea{ça, on verra}}

The decidability of the emptiness in \cref{theorem:decidability-LU} does not necessarily mean that the exact synthesis can be achieved.
In fact, we show in the following result that the consistency-synthesis for L/U-\PIPROBTAs{} is intractable in practice, as the set of valuations cannot be represented using, \eg{} a finite union of polyhedra.

\begin{proposition}\label{proposition:intractability-LU}
	The result of the consistency-synthesis for L/U-\PIPROBTAs{} cannot be represented using any formalism for which the emptiness of the intersection is decidable.
\end{proposition}
\begin{proof}
	We adapt a reasoning from \cite{JLR15} originally showing that the synthesis for L/U-PTAs is intractable.
	Let us assume an arbitrary \PIPROBTA{} (not necessarily L/U).
	For each parameter~$\param_i$, let us create a parameter~$\param_i^l$ and a parameter~$\param_i^u$ (and delete $\param_i$).
	Then, let us replace
		each constraint $\clock \leq \param_i$ with $\clock \leq \param_i^u$,
		each constraint $\clock < \param_i$ with $\clock < \param_i^u$,
		each constraint $\clock \geq \param_i$ with $\clock \geq \param_i^l$,
		each constraint $\clock > \param_i$ with $\clock > \param_i^l$,
		and
		each constraint $\clock = \param_i$ with $\clock \geq \param_i^l \land \clock \leq \param_i^u$.
	We obtain an L/U-\PIPROBTA{}.
	Clearly, if $\param_i^l = \param_i^u$ for all~$i$, then the behavior of the L/U-\PIPROBTA{} is identical to that of the original \PIPROBTA{}.
	
	Now, assume that the result of the consistency-synthesis for L/U-\PIPROBTAs{} can be represented using a formalism for which the emptiness of the intersection is decidable.
	We can therefore synthesize all valuations for which the L/U-\PIPROBTAs{} is consistent using such a formalism.
	Then, let us intersect this result with $\bigwedge_{1 \leq i \leq \ParamCard} \param_i^l = \param_i^u$.
	Finally, let us check whether this intersection is empty.
	We are thus able to test the consistency-emptiness of the original \PIPROBTA{}---which contradicts~\cref{theorem:undecidability}.
\end{proof}

In the rest of the section, despite the negative results of \cref{theorem:undecidability,proposition:intractability-LU}, we will still attempt to address synthesis for the full class of \PIPROBTAs{}.

\subsection{A symbolic semantics for \PIPROBTAs{}}\label{ss:symbolic-semantics}

We equip \PIPROBTAs{} with a symbolic semantics, defined below.
Basically, it is inline with the symbolic semantics defined for parametric timed automata (see \eg{} \cite{ACEF09,JLR15}), with the addition of probabilistic intervals on the edges; as a consequence, the semantics becomes not an LTS, but an \IMDP{}.
Remark that this is a conservative extension of the symbolic semantics of \IPROBTA{} presented in \cref{definition:IPTA:symbolic-semantics}.

\begin{definition}[Symbolic semantics of a \PIPROBTA{}]
	Given a \PIPROBTA{} $\piprobta = (\Actions, \Loc, \locinit, \Clock, \Param, \TransitionsIPTA)$,
	the symbolic semantics of $\piprobta$ is given by the \IMDP{} $(\SymbStates, \symbstateinit, \TransitionsIPTA, \TransitionsMDP)$, with
	\begin{itemize}
		\item $\SymbStates = \{ (\loc, \C) \in \Loc \times \PZones \}$,  %
			$\symbstateinit = (\locinit, \timelapse{(\bigwedge_{1 \leq i\leq\ClockCard}\clock_i=0)} )$, %
		\item $((\loc, \C), \pedge, \DistIPTA') \in \TransitionsMDP$ if there exists $\pedge = (\loc,\guard,\action,\DistIPTA) \in \TransitionsIPTA$ such that for all $\loc' \in \Loc$, for all $\resets \subseteq \Clock$ such that $\DistIPTA(\resets, \loc') > 0$, 
		$\C' = \timelapse{\big(\reset{\C \land \guard}{\resets}\big )} $, %
		and $\DistIPTA'((\loc', \C')) = \DistIPTA(\resets, \loc')$.
	\end{itemize}
\end{definition}

Observe that, whenever a \PIPROBTA{} has no probabilistic choice (\ie{} is a PTA), then the \IMDP{} becomes a labeled transition system, and the symbolic semantics matches that of parametric timed automata.
We refer to the symbolic semantics of~$\piprobta$ as the \emph{parametric probabilistic zone graph} of~$\piprobta$.

Just as in parametric timed automata, the number of symbolic states in a \PIPROBTA{} can be infinite in general.

In parametric timed automata, the \emph{reachability condition} is the projection onto the parameters of a parametric zone (see~\cite{JLR15}).
It is well-known that, given a symbolic run of a parametric timed automaton leading to a symbolic state $(\loc, \C)$, there exists an equivalent concrete run iff $\param \models \projectP{\C}$ (see \eg{}~\cite{HRSV02}).
Since our definition of zones matches that of~\cite{HRSV02}, this results extends to \PIPROBTAs{} in a straightforward manner.

\begin{lemma}\label{lemma:HRSV02}
	Let~$\piprobta$ be a \PIPROBTA{}.
	Consider a run in the parametric probabilistic zone graph of~$\piprobta$ reaching state~$(\loc, \C)$.
	Let $\pval$ be a parameter valuation.
	Then, there exists an equivalent run in $\valuate{\piprobta}{\pval}$ iff $\pval \models \projectP{\C}$.
\end{lemma}

By equivalent run, we mean (just as for parametric timed automata) an identical discrete structure (locations and edges).

\begin{table}
\centering
\footnotesize

	\begin{tabular}{| c | c | c| c|}
		\hline
		\cellHeader{State} & \cellHeader{Location} & \cellHeader{$\C$} & \cellHeader{$\projectP{\C}$} \\
		\hline
		$\symbstate_0$ & $\loc_0$ & $x = y \land x \geq 0 \land \param \geq 0$ & $\param \geq 0$ \\
		\hline
		$\symbstate_1$ & $\loc_1$ & $0 \leq x - y < 2 \land y \geq 0 \land \param \geq 0$ & $\param \geq 0$ \\
		\hline
		$\symbstate_2$ & $\loc_2$ & $0 \leq y - x < 2 \land x \geq 0 \land \param \geq 0$ & $\param \geq 0$ \\
		\hline
		$\symbstate_3$ & $\loc_3$ & $2 \leq x - y \leq \param \land y \geq 0 $ & $\param \geq 2$ \\
		\hline
		$\symbstate_4$ & $\loc_4$ & $x = y \land x \geq 0 \land \param \geq 2$ & $\param \geq 2$ \\
		\hline
		$\symbstate_5$ & $\loc_5$ & $0 \leq y - x \leq 1 \land x \geq 1 \land \param \geq 0$ & $\param \geq 0$ \\
		\hline
		$\symbstate_6$ & $\loc_2$ & $1 \leq y - x \leq 2 \land x \geq 0 \land \param \geq 0$ & $\param \geq 0$ \\
		\hline
		$\symbstate_7$ & $\loc_5$ & $y \geq 2 \land y = x + 1 \land \param \geq 0$ & $\param \geq 0$ \\
		\hline
		$\symbstate_8$ & $\loc_2$ & $y \geq 2 \land y = x + 2 \land \param \geq 0$ & $\param \geq 0$ \\
		\hline
	\end{tabular}
	
	\caption{Description of the states in \cref{figure:example:IMDP}}
	\label{table:pzones:description}
\end{table}
\begin{example}
	The parametric probabilistic zone graph of the \PIPROBTA{} in\linebreak[4] \cref{figure:example:PIPROBTA} is the \IMDP{} given in \cref{figure:example:IMDP}.
	The symbolic states $\symbstate_i = (\loc_i, \C_i)$ are expanded in \cref{table:pzones:description}.
	In addition, we also give the reachability condition of each state, \ie{} the projection onto the parameters of the zone ($\projectP{\C}$).
\end{example}
\subsection{A construction for consistency-synthesis for \PIPROBTAs{}}\label{ss:synthesis}

Unlike for \IPROBTAs{} / \IMDPs{} where inconsistent states can only be
avoided by enforcing their incoming probabilities to $0$, there are
two ways of avoiding inconsistent states in \PIPROBTAs{}. Indeed, while
imposing a $0$~probability to all transitions going to inconsistent
states is a safe choice, it is also possible to avoid inconsistent
states by cleverly choosing parameter values such that the guards of
transitions potentially going to these states are never satisfied.

The construction we propose for synthesizing parameter valuations
ensuring consistency of a given \PIPROBTA{} is based on the following
observation: Since parameters only occur in transition guards, the
choice of parameter values cannot interfere with the choice of
probability distributions matching (or not) the specified
intervals.
That comes from the fact that, given a state~$\symbstate$, all successors of this state via a given transition have the same parameter constraint (this would not hold with invariants).
As a consequence, states that can be made unreachable
through probabilistic choice can be made so regardless of the choice
of parameter values.

\paragraph{Notations}
We first introduce a few notations to make our construction more compact.
Given an \IMDP{}~$\imdp = (\States, \stateinit, \TransitionsIPTA, \TransitionsMDP)$ (representing the semantics of a \PIPROBTA{}~$\piprobta$),
let $\TOut(\symbstate)$ denote the set of transitions of source~$\symbstate$, \ie{} $\TOut(\symbstate)=\{ (\symbstate,\pedge, \DistIMDP) \in \TransitionsMDP\}$.

Given a transition $(\symbstate,\pedge, \DistIMDP) \in \TransitionsMDP$, we may want to forbid this transition; recall that the guard (in the original \PIPROBTA{}) is the same for all targets, as there is a single guard per interval distribution.
As we have no invariants, all target states of a given transition have the same reachability condition (\ie{} $\projectP{\C'}$, for a target $\symbstate' = (\loc', \C')$).
Therefore, in order to forbid a transition, it suffices to negate the reachability condition of any of the target states of this transition.
Let $\ForbidGuard(\DistIMDP)$ denote this result, \ie{} $\ForbidGuard(\DistIMDP) = \neg \projectP{\C'}$, where $(\loc', \C')$ is an (arbitrary) target state of~$\DistIMDP$.

Finally, recall that a disjunction over an empty set of clauses is by definition false. Therefore, we use $\bigvee_i \K_i$ to denote the union over a set that returns the usual union of~$\K_i$ for all $i$ in the set if the set is non-empty, or $\KFalse$ if the set is empty. Similarly, $\bigwedge_i \K_i$ denotes the intersection over a set that returns the usual intersection of~$\K_i$ for all $i$ in the set if the set is non-empty, or $\KTrue$ if the set is empty.

\begin{example}
	Consider the \IMDP{} in \cref{figure:example:IMDP}, which is
        the zone graph of the \PIPROBTA{} from
        \cref{figure:example:PIPROBTA}. Recall that the description of
        the symbolic states of the \IMDP{} from
        \cref{figure:example:IMDP} is given in \cref{table:pzones:description}.
        We illustrate the constructions for $\TOut$ and $\ForbidGuard$ given above.

Clearly, there is only one outgoing transition from state $\symbstate_1$, which is labeled with $\pedge_2$. As a consequence, we have $\TOut(\symbstate_1) = (\symbstate_1,\pedge_2,\DistIMDP)$ with $\DistIMDP(\symbstate_3) = [0,0.2]$, $\DistIMDP(\symbstate_4) = [0,0.3]$, and $\DistIMDP(\symbstate_i) = [0,0]$ for $i \notin \{3,4\}$.

Remark that, as explained above, all the states that are reachable through $\DistIMDP$ have the same reachability condition (given in \cref{table:pzones:description}). As a consequence, we have $\ForbidGuard(\DistIMDP) = \lnot (\param \ge 2) \equiv \param < 2$.
\end{example}

We now propose a characterization of the set of parameter valuations
that ensure consistency of a given \PIPROBTA{} under the assumption
that its parametric probabilistic zone graph is finite.

Let $\piprobta$ be a \PIPROBTA{}, and let~$\imdp$ be its parametric probabilistic zone graph. Assume $\imdp$ is finite with state space $\SymbStates = \{\symbstate_0,\dots,\symbstate_n\}$.
Consider the formula $cons(c_{s_0},...,c_{s_n})$ defined as:

\[(c_{s_0}=\top) \wedge  \bigwedge_{\symbstate \in \SymbStates}
\bigwedge_{(\symbstate,\pedge, \DistIMDP) \in \TOut(\symbstate) } 
\left( \neg c_\symbstate\vee \ForbidGuard(\DistIMDP) \lor \bigvee_{\States' \in \CombiInt(\DistIMDP)}  \bigwedge_{\symbstate' \in \States' \setminus\{s\}}c_{\symbstate'}\right)\text{.}\]

Intuitively in this formula the variable $c_\symbstate$ represents whether
state $\symbstate$ can be reachable in an implementation. Recall that this can
only be true if $\symbstate$ is consistent. As a consequence, the formula can only be true
when the valuation of the parameters is coherent with the
consistent states. Indeed, this formula ensures that the initial state
is reachable and that, for any state~$\symbstate$ and any outgoing transition of this state, either:
\begin{itemize}
	\item the source state $\symbstate$ is not reachable ($\neg c_\symbstate$), or
	\item the transition is disabled due to the valuation of the parameters ($\ForbidGuard(\DistIMDP)$), or
	\item the transition is enabled and thus there must exist a feasible support for which all reachable states are also consistent. 
\end{itemize}

The set of all solutions for the consistency synthesis problem is thus
given as the set of solutions (in terms of parameter valuations) of
the equation:
\begin{equation}\label{eq:conssynt}
\bigvee_{(c_{\symbstate_0},\dots,c_{\symbstate_n})\in \{\top,\bot\}^{n+1}} cons(c_{s_0},\dots,c_{\symbstate_n})
\end{equation}

In the following, this procedure (\ie{} solving equation
(\ref{eq:conssynt})) is called $\PIPTAsynth$.
The intuition behind procedure $\PIPTAsynth$ is that we
``guess'' the states that will be present in the implementation
through the first disjunction (states for which $c_\symbstate
= \top$), and then verify using $cons(c_{s_0},\dots,c_{\symbstate_n})$
that the resulting implementation is well-defined.

Obviously, the empty set of parameter valuations is always a solution
to equation (\ref{eq:conssynt}). Indeed, in this case, one can set
$c_{\symbstate_0} = \top$ and $c_{\symbstate} = \bot$ for all
$\symbstate \ne \symbstate_0$ and then $\ForbidGuard(\DistIMDP)$ is
true for all outgoing transitions of $\symbstate_0$. If this is the
only solution, then the \PIPROBTA{} $\piprobta$ is
inconsistent. Otherwise, $\piprobta$ is consistent.

We first illustrate our construction on an example and then show that
it is sound and complete when the parametric probabilistic zone graph
of $\piprobta$ is finite.

\begin{example}
  \label{example:PIPTAsynth}
	We now apply our construction to the IMDP
	from \cref{figure:example:IMDP}. Recall that parameters
	are non-negative, therefore $\param < 0 \equiv \KFalse$.
	First observe that either $\symbstate_1$
        has to be non-reachable ($c_{\symbstate_1} = \bot$) or its
        outgoing transition needs to be forbidden, because there is no
        set $\States' \in \CombiInt(\DistIMDP)$. As a consequence, we
        obtain the following constraint: $\neg c_{\symbstate_1}
        \vee (\param < 2)$. There are no constraints for states
        $\symbstate_5, \symbstate_7$ and $\symbstate_8$ as they have
        no outgoing transitions. For state $\symbstate_6$, we have the
        following constraint:

        $$\neg c_{\symbstate_6} \vee (\param < 0) \vee
        c_{\symbstate_7} \vee (c_{\symbstate_7} \land
        c_{\symbstate_8}) \equiv \neg c_{\symbstate_6} \vee (\param
        <0) \vee c_{\symbstate_7} \equiv \neg c_{\symbstate_6} \vee c_{\symbstate_7}$$

        Similarly, for state $\symbstate_2$, we obtain:

        $$\neg c_{\symbstate_2} \vee (\param
        <0) \vee c_{\symbstate_5} \equiv \neg c_{\symbstate_2} \vee c_{\symbstate_5}$$

        Finally, state $\symbstate_0$ yields the following:

        $$\neg c_{\symbstate_0} \vee (\param
        <0) \vee c_{\symbstate_1} \equiv \neg c_{\symbstate_0} \vee c_{\symbstate_1}$$

        Clearly, when put together in equation (\ref{eq:conssynt}), we
        obtain the following (after simplifications):

        $$\bigvee_{(c_{\symbstate_0},\dots,c_{\symbstate_n})\in
          \{\top,\bot\}^{n+1}} (c_{\symbstate_0}) \land (c_{\symbstate_1}) \land (\param < 2)
        \land (\neg
        c_{\symbstate_2} \vee c_{\symbstate_5}) \land (\neg
        c_{\symbstate_6} \lor c_{\symbstate_7})$$

        The solutions are therefore all parameter valuations such that
        $\param < 2$, and can be obtained for all assignments of
        $c_{\symbstate}$ such that $c_{\symbstate_0} =
        c_{\symbstate_1} = \top$, $c_{\symbstate_2} \Rightarrow
        c_{\symbstate_5}$ and $c_{\symbstate_6} \Rightarrow
        c_{\symbstate_7}$.

\end{example}

We now prove that our construction is indeed correct whenever the
parametric probabilistic zone graph of the given \PIPROBTA{}
$\piprobta$ is finite.

\begin{proposition}[Correctness]
	Let $\piprobta$ be a \PIPROBTA{}, and let~$\imdp$ be its parametric probabilistic zone graph.
	Assume $\imdp$ is finite.
	Assume that the set of parameter valuations satisfying
        equation (\ref{eq:conssynt}) is not empty and let $\pval$ be
        such a parameter valuation.
	
	Then $\valuate{\piprobta}{\pval}$ is consistent.
\end{proposition}
\begin{proof}
Let $\pval$ be a solution of equation (\ref{eq:conssynt}). As a
consequence, there must exist an assignment
$\kappa_{s_0},...,\kappa_{s_n}$ of the variables $c_\symbstate$ such
that $\pval(cons(\kappa_{s_0},\dots,\kappa_n))$. Moreover, for each state $\symbstate$ such that $\kappa_\symbstate = \top$, the following equation is satisfied:

$$ \bigwedge_{(\symbstate,\pedge, \DistIMDP) \in \TOut(\symbstate) } 
\left( \pval(\ForbidGuard(\DistIMDP)) \lor \bigvee_{\States' \in \CombiInt(\DistIMDP)}  \bigwedge_{\symbstate' \in \States' \setminus\{s\}}\kappa_{\symbstate'}\right)$$

Therefore, for each
$(\symbstate,\pedge, \DistIMDP) \in \TOut(\symbstate)$, either
$\pval(\ForbidGuard(\DistIMDP))$ is true (in this case the transition
cannot be taken due to timing parameters and is therefore absent from
$\pval(\imdp)$), or there exists a distribution $\DistMDP$ matching
$\DistIMDP$ such that all states $\symbstate'$ such that
$\DistMDP(\symbstate') >0 $ are such that $\kappa_{\symbstate'}
= \top$.

We can therefore construct an MDP whose states are exactly the states
$\symbstate$ such that $\kappa_{\symbstate} = \top$, and whose
transitions are given the distributions $\DistMDP$ defined above, that
clearly satisfies the IMDP $\pval(\imdp)$.

By \cref{theorem:IPTA:consistency}, we can therefore conclude that $\piprobta$ is consistent.
\end{proof}

We now show that our construction is complete whenever the
parametric probabilistic zone graph of the given \PIPROBTA{}
$\piprobta$ is finite.

\begin{proposition}[Completeness]
	Let $\piprobta$ be a \PIPROBTA{}, and let~$\imdp$ be its
	parametric probabilistic zone graph.  Assume $\imdp$ is
	finite.  Let $\pval$ be such that $\valuate{\piprobta}{\pval}$
	is consistent. Then $\pval$ is a solution of equation
	(\ref{eq:conssynt}).
\end{proposition}
\begin{proof}
  Since $\valuate{\piprobta}{\pval}$ is consistent, there must exist,
  by \cref{theorem:IPTA:consistency,lemma:structure}, an \MDP{} $\mdp$ with the same
  structure as $\pval(\imdp)$ that satisfies $\pval(\imdp)$.

  We now propose a valuation of the variables $c_{\symbstate}$ and
  show that, for this valuation, equation (\ref{eq:conssynt}) is true.

  For all state $\symbstate$ in $\imdp$, let $c_{\symbstate} = \top$
  if $\symbstate$ is reachable (and present) in $\mdp$, and $\bot$ otherwise. We now
  show that the equation
  $\pval(cons(c_{\symbstate_0},\dots,c_{\symbstate_n}))$ is
  true. Clearly, we have $c_{\symbstate_0} = \top$, so we just have to
  show that for all state $\symbstate$ in $\imdp$, 
  $$\bigwedge_{(\symbstate,\pedge, \DistIMDP) \in \TOut(\symbstate) } 
\left( \neg c_\symbstate\vee \pval(\ForbidGuard(\DistIMDP)) \lor
\bigvee_{\States' \in \CombiInt(\DistIMDP)}  \bigwedge_{\symbstate'
  \in \States' \setminus\{s\}}c_{\symbstate'}\right) = \mathsf{true}$$

If $\symbstate$ is such that $c_{\symbstate} = \bot$, then this is
trivial. Otherwise, let $(\symbstate,\pedge, \DistIMDP) \in
\TOut(\symbstate)$. Clearly, since the transition is present in
$\pval(\imdp)$, we have $\pval(\ForbidGuard(\DistIMDP)) =
\mathsf{false}$. Moreover, since $\symbstate$ is present (and
reachable) in $\mdp$, there is a transition $(\symbstate,
\pedge,\DistMDP)$ in $\mdp$ such that $\DistMDP \preceq_{\RelSimMDP}
\DistIMDP$ for the witnessing relation $\RelSimMDP$. As a consequence,
the set $\States' = \{\symbstate' \ \mid \ \DistMDP(\symbstate') >0\}$ is
such that $\States' \in \CombiInt(\DistIMDP)$. Moreover, all states in
$\States'$ are reachable by
construction, thus $\symbstate' \in \States' \Rightarrow c_{\symbstate'} =
\top$. Therefore,
$\pval(cons(c_{\symbstate_0},\dots,c_{\symbstate_n}))$ is true.
\end{proof}

\begin{remark}
Our construction is based on the parametric probabilistic zone
graph. It is sound and complete when this zone graph is
finite. However, the resulting equation contains infinite conjunctions
and disjunctions when the parametric probabilistic zone graph is
infinite, rendering it useless in practice in this case.

However, in practice, one could truncate the parametric probabilistic
zone graph up to a certain depth, which would allow computing an
approximation of the set of parameter valuations ensuring consistency.
\end{remark}

\bd{Exemple viré, j'ai illustré la construction directement sur le
  PIPTA du début.}
\subsection[Parametric Consistent Reachability]{Parametric Consistent Reachability}\label{ss:consistent-reachability}

A model that is inconsistent is a model that can be considered as ill-formed; therefore, synthesizing valuations for a model to be consistent is an important problem.
However, it may not be seen as the final problem a system designer aims at solving.
More common problems are reachability, safety, unavoidability, or more complex properties expressed, \eg{} using logic formulas.

In this section, we illustrate how consistency synthesis can be combined with existing synthesis algorithms.
As a proof of concept, we consider the following parametric consistent reachability synthesis problem:

\defProblem
	{parametric consistent reachability synthesis}
	{A \PIPROBTA{}~$\piprobta$, a set of goal locations~$\LocGoal$}
	{find all parameter valuations $\pval$ for which $\valuate{\piprobta}{\pval}$ is consistent and at least one location in~$\LocGoal$ is reachable in $\valuate{\piprobta}{\pval}$.}

The corresponding emptiness problem, \ie{} the emptiness of the valuation set for which a \PIPROBTA{} is consistent and at least one goal location is reachable, is clearly undecidable: it suffices to consider a \PIPROBTA{} with no probabilities.
This gives a PTA, for which reachability emptiness is undecidable: so, clearly, a PTA is always consistent and therefore consistent reachability emptiness reduces to reachability emptiness, which is undecidable, as shown in~\cite{AHV93}.

Still, we will propose a method to perform parametric consistent reachability synthesis for \PIPROBTA{}; again, this method only works when the parametric probabilistic zone graph is finite.

First, let us rule out the following naive method.  We could have
considered the PTA obtained from a~\PIPROBTA{} by removing all
probabilities, then we could have synthesized valuations for which
reachability of location~$\LocGoal$ is ensured, which can be obtained
using the algorithm described in, \eg{} \cite{JLR15}, and that we will
call~\EFsynth{}.  This gives a constraint~$\Kreach$.  Then, we could
synthesize the constraint~$\Kcons$ obtained from $\PIPTAsynth$.
Finally, we could have considered the intersection
$\Kreach \land \Kcons$. However this is not satisfactory (and wrong),
as shown in the example below.

\begin{example}
	Consider the \PIPROBTA{} in \cref{figure:counter-example-EF}.
	Assume $\LocGoal = \{ \loc_1 \}$.
	Clearly, $\loc_1$ is inconsistent, as its successor has an interval distribution that admits no implementation.
	$\loc_1$ can easily be discarded by assigning it a 0-probability from~$\loc_0$ while keeping the interval consistent.
	
	On this \PIPROBTA{} without probabilities, \EFsynth{} will
	output $\KTrue$ as any parameter valuation may
	reach~$\loc_1$.  \PIPTAsynth{} will also output $\KTrue$.  The
	intersection gives~$\KTrue$, while the set of valuations for
	which $\loc_1$ is reachable and the system is consistent is
	empty.
\end{example}
\begin{figure}
	
	{\centering
		
	\newcommand{\hscale}{2}
	\newcommand{\vscale}{1.5}
	\begin{tikzpicture}[node distance=2cm, auto, ->, >=stealth']
		\node[location0, initial] (l0) at (0,0) {$\loc_0$};
		\node [probchoice] (l0choice) at (1*\hscale, 0*\vscale) {};

		\node[location1] (l1) at (2*\hscale, +.7*\vscale) {$\loc_1$};
		\node[location2] (l2) at (2*\hscale, 0) {$\loc_2$};
		\node[location3] (l3) at (2*\hscale, -.7*\vscale) {$\loc_3$};

		\node [probchoice] (l1choice) at (3*\hscale, .7*\vscale) {};

		\node[location4] (l4) at (4*\hscale, +.7*\vscale) {$\loc_4$};

		\path
			(l0) edge node[below]{\begin{tabular}{c}\footnotesize $\clock \geq \param$ \\ $a$\end{tabular}} (l0choice) %
		
			(l0choice) edge[probedge] node[above, xshift=-5]{\footnotesize \begin{tabular}{c}$[0, 0.6]$\end{tabular}} (l1)
			(l0choice) edge[probedge] node[above, xshift=5]{\footnotesize $[0, 0.6]$} (l2)
			(l0choice) edge[probedge] node[below]{\footnotesize \begin{tabular}{c}$[0, 0.9]$\end{tabular}} (l3)
			
			(l1) edge node[below]{} (l1choice)
			(l1choice) edge[probedge] node[below]{\footnotesize $[0.2, 0.4]$} (l4)
		;

		\AngleA{l0choice}{l1}{l3};

		\end{tikzpicture}
	
	}

	\caption{A \PIPROBTA{} for which no valuation allows for consistent reachability of~$\loc_1$}
	\label{figure:counter-example-EF}
\end{figure}
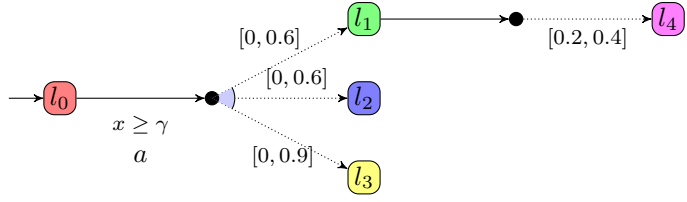

\bd{Etienne : relis stp}
We propose the following construction, which we adapt from our
construction \PIPTAsynth{}. Recall that in this construction, we use for each state $\symbstate$
 a variable $c_{\symbstate}$ that encodes the potential presence of state
$\symbstate$ in an implementation (and therefore imposes that this
state is consistent). Unfortunately, the presence of such a state in
an implementation is not sufficient to ensure that this state is
reachable from the initial state. In order to guarantee reachability,
we therefore have to add variables and constraints to
\PIPTAsynth{}.

We therefore add new variables $r_{\symbstate}$ for all states
$\symbstate$ in the parametric probabilistic zone graph. These
variables will be assigned values in $[0,N] \cup \{\infty\}$, where $N$ is the total
number of states of the parametric probabilistic zone graph.

Then, in order to ensure reachability of the goal locations, we add
the following constraints:

\begin{itemize}
\item $r_{(\loc,\C)} = 0 \iff \loc \in \LocGoal$
\item for all states $\symbstate = (\loc,\C)$ in the parametric
  probabilistic zone graph such that $\loc \notin \LocGoal$, we impose that
  either $(r_{\symbstate} = \infty)$ or

  $$\bigvee_{(\symbstate,\pedge, \DistIMDP)
    \in \TOut(\symbstate)} \left (\neg \ForbidGuard(\DistIMDP) \land
  \bigvee_{\States' \in \CombiInt(\DistIMDP)} \left (
  \bigwedge_{\symbstate' \in \States' \setminus\{s\}}c_{\symbstate'}
  \wedge \bigvee_{\symbstate' \in \States'} (r_{\symbstate} =
  r_{\symbstate'} +1)\right )  \right)$$
\end{itemize}

Now, solving the conjunction of equation (\ref{eq:conssynt}) from
\PIPTAsynth{} and the constraints presented above, while imposing that
$r_{\symbstate_0} < \infty$ will yield exactly the set of parameter
valuations ensuring the consistent reachability of goal locations from
$\LocGoal$. We call this new procedure \PIPTAreachsynth{}.

\begin{proposition}
	Let $\pval$ be a parameter valuation satisfying the result of\linebreak[4] \PIPTAreachsynth{}. %
	Then, $\valuate{\piprobta}{\pval}$ is consistent and at least one location in~$\LocGoal$ is reachable in $\valuate{\piprobta}{\pval}$.
\end{proposition}
\begin{proof}
  First observe that any parameter valuation $\pval$ obtained through
  \PIPTAreachsynth{} needs to satisfy \PIPTAsynth{}. As a consequence,
  $\valuate{\piprobta}{\pval}$ is consistent.  Moreover, the
  additional constraints provided in \PIPTAreachsynth{} ensure that
  whenever $r_{\symbstate} < \infty$, there is an execution of length at
  most $r_{\symbstate}$ from $\symbstate$ to a state $(\loc,\C)$ such
  that $\loc \in \LocGoal$ in the parametric probabilistic zone graph
  of $\valuate{\piprobta}{\pval}$. Since we impose that
  $r_{\symbstate_0} < \infty$, $\LocGoal$ is indeed reachable in
  $\valuate{\piprobta}{\pval}$.
\end{proof}
\begin{example}
	Let us come back to \cref{figure:counter-example-EF}.
	\PIPTAsynth{} yields the entire set of parameter
        valuations. By construction, the parametric probabilistic zone
        graph of this \PIPROBTA{} is almost identical to the
        \PIPROBTA{} itself (the states will be $\symbstate_i = (\loc_i, \param \ge
        0)$ for all $i$). However, in \PIPTAreachsynth{}, it will be
        impossible to set $r_{\symbstate_1}$ to a finite value as
        the feasible support of its outgoing transition is empty. As a
        consequence, \PIPTAreachsynth{} will yield the empty set of
        parameter valuations, as expected.
\end{example}
\begin{example}
  Assume now that we would like to synthesize the parameter values
  ensuring the consistent reachability of $\loc_5$ in the \PIPROBTA{}
  given in \cref{figure:example:PIPROBTA}. Recall that the
  probabilistic zone graph is given in \cref{figure:example:IMDP} and
  the solutions of \PIPTAsynth{} are given in
  \cref{example:PIPTAsynth}. In the process of solving
  \PIPTAreachsynth{} on this example, we are required to set
  $r_{\symbstate_i} = \infty$ for $i \in \{1,3,4,8\}$. We also have to
  set $r_{\symbstate_5} = r_{\symbstate_7} = 0$. Finally, in
  addition to the above constraints and those obtained in
  \PIPTAsynth{}, \PIPTAreachsynth{} yields the following:

  \begin{itemize}
  \item For state $\symbstate_6$: either $(r_{\symbstate_6} = \infty)$, or
    $$(\gamma \ge 0) \land
    ((c_{\symbstate_7} \land (r_{\symbstate_6} = r_{\symbstate_7}+1))
    \lor (c_{\symbstate_7} \land c_{\symbstate_8}\land
    ((r_{\symbstate_6} = r_{\symbstate_7}+1) \lor (r_{\symbstate_6} =
    r_{\symbstate_8}+1))))$$

  \item For state $\symbstate_2$: either $(r_{\symbstate_2} = \infty)$, or
    $$(\gamma \ge 0) \land
    ((c_{\symbstate_5} \land (r_{\symbstate_2} = r_{\symbstate_5}+1))
    \lor (c_{\symbstate_5} \land c_{\symbstate_6}\land
    ((r_{\symbstate_2} = r_{\symbstate_5}+1) \lor (r_{\symbstate_2} =
    r_{\symbstate_6}+1))))$$

  \item For state $\symbstate_0$: either $(r_{\symbstate_0} = \infty)$, or
    $$(\gamma \ge 0) \land
    ((c_{\symbstate_1} \land (r_{\symbstate_0} = r_{\symbstate_1}+1))
    \lor (c_{\symbstate_1} \land c_{\symbstate_2}\land
    ((r_{\symbstate_0} = r_{\symbstate_1}+1) \lor (r_{\symbstate_0} =
    r_{\symbstate_2}+1))))$$
  \end{itemize}

In the end, we can set $(r_{\symbstate_6} = 1), (r_{\symbstate_2} =
1)$, and $(r_{\symbstate_0} = 2)$ for instance. In this case, we still
obtain the same set of parameter valuations as in \PIPTAsynth{}:
all those satisfying $(\param < 2)$.
\end{example}

\begin{remark}\label{remark:acyclic}
	Observe that, for acyclic \PIPROBTAs{} (\ie{} the underlying graph of which contains no cycle), the answer to the parametric consistent reachability synthesis problem can be effectively computed.
	Indeed, the procedure presented above consists in a procedure to be solved on a set of states.
	If that set is finite, the procedure can be effectively solved with an exact result.
	
	This result can also extended to \PIPROBTAs{} the symbolic semantics of which is acyclic (\ie{} the underlying \IMDP{} contains no cycle).
	However, it may not be possible to decide whether an arbitrary \PIPROBTAs{} has a finite symbolic semantics.
\end{remark}

\section{Conclusion}\label{section:conclusion}

In this work, we provided abstractions to reason on systems involving real-time constraints and probabilities:
first, by allowing probabilities to range in some intervals,
and, second, by allowing timing constants to be abstracted in the form of parameters.
Without parameters, we proposed an approach to decide whether an interval probabilistic timed automaton is consistent, \ie{} admits an implementation based on a simulation relation.
When adding parameters, the mere existence of a parameter valuation yielding consistency is undecidable.
However, when the set of parameters is partitioned between lower-bound parameters and upper-bound parameters, decidability is ensured.
We also proposed a procedure to synthesize valuations ensuring
consistency for \PIPROBTAs{} whose parametric probabilistic zone graph
is finite, as well as to ensure consistent reachability.

\paragraph{Discussion}\label{discussion:consistency}
We believe our definition of consistency allows for incremental design: one can first define range for probabilities and range for timing parameters.
Then, depending on refined design choices, one will assign interval probabilities with punctual values, and valuate timing parameters.
Clearly, inconsistent probabilistic distributions can be seen as ill-formed models---just as deadlocks, for examples.
One could argue that, contrarily to deadlocks, one could statically detect such situations, or even forbid them statically.
However, we see two reasons not to do so.
First, we believe that allowing these situations could be used as an additional freedom, that can be then detected and corrected using the methods described in this manuscript.
That is, inconsistent intervals do not need to be removed statically if there is another way to remove them (using other probabilities or timing parameters).
Second, our work builds on top on works where \emph{parametric} probabilistic bounds can be used (\eg{} \cite{BDSyncop15,DBLP:conf/vmcai/DelahayeLP16}).
In this latter case, the static detection does not work.
As our ultimate goal is to reintroduce parametric intervals in the future (see below), we believe our definition of consistency is worth exploring.

\paragraph{Future works}
We envision several future works.
First, exhibiting subclasses of \PIPROBTAs{} for which exact synthesis can be achieved is on our agenda.
As the use of timing parameters seems critical in our undecidability results, relying on recent works exhibiting decidable subclasses of parametric timed automata, such as bounded integer parameters (see~\cite{JLR15}) or reset-parametric timed automata (see~\cite{ALR16ICFEM,ALR18ACSD}), can serve as a first basis for a probabilistic extension.

Finally, we are interested in considering higher-level abstractions of probabilities; notably, using parameters instead of intervals with constant bounds (as in \cite{DBLP:conf/vmcai/DelahayeLP16} for parametric interval Markov chains) is of high interest, and makes the notion of consistency even more delicate, as tuning the parametric bounds in an interval may impact the consistency of other probabilistic distributions.

\section*{Acknowledgements}
We would like to thank the anonymous reviewers for useful comments that helped us to improve the manuscript.

\ShortVersion{
\section*{Bibliography}
}

\newcommand{\CCIS}{Communications in Computer and Information Science}
\newcommand{\ENTCS}{Electronic Notes in Theoretical Computer Science}
\newcommand{\FAC}{Formal Aspects of Computing}
\newcommand{\FI}{Fundamenta Informaticae}
\newcommand{\FMSD}{Formal Methods in System Design}
\newcommand{\IJFCS}{International Journal of Foundations of Computer Science}
\newcommand{\IJSSE}{International Journal of Secure Software Engineering}
\newcommand{\IPL}{Information Processing Letters}
\newcommand{\JLAP}{Journal of Logic and Algebraic Programming}
\newcommand{\JLAMP}{Journal of Logical and Algebraic Methods in Programming} %
\newcommand{\JLC}{Journal of Logic and Computation}
\newcommand{\LMCS}{Logical Methods in Computer Science}
\newcommand{\LNCS}{Lecture Notes in Computer Science}
\newcommand{\RESS}{Reliability Engineering \& System Safety}
\newcommand{\STTT}{International Journal on Software Tools for Technology Transfer}
\newcommand{\TCS}{Theoretical Computer Science}
\newcommand{\ToPNoC}{Transactions on Petri Nets and Other Models of Concurrency}
\newcommand{\TSE}{{IEEE} Transactions on Software Engineering}

\ifdefined\VersionLong

	\renewcommand*{\bibfont}{\small}
	\printbibliography[title={References}]
\else

	\bibliographystyle{elsarticle-harv} 
	\bibliography{PIPTA}

\fi
\ifdefined\WithReply
	\input{letter2.tex}
\fi

\end{document}